\documentclass[letterpaper, 10 pt, conference]{ieeeconf}  

\IEEEoverridecommandlockouts                              

\usepackage{graphicx}
\usepackage{tikz,tikz-3dplot}
\usepackage{pgfplots} 
\usepackage{pgfgantt}
\usepackage{pdflscape}
\pgfplotsset{compat=newest} 
\pgfplotsset{plot coordinates/math parser=false}
\pgfplotsset{table/search path={figures/data}}
\usepackage{threeparttable}
\usetikzlibrary{arrows,arrows.meta}
\usepackage[labelformat=simple]{subcaption}
\usepackage{import}
\DeclareCaptionLabelSeparator{periodspace}{.\quad}
\usepackage{amsmath}
\allowdisplaybreaks[1]
\usepackage{amssymb}
\usepackage{wasysym}
\usepackage{mathtools}
\usepackage{xfrac}

\newtheorem{assumption}{Assumption}
\newtheorem{theorem}{Theorem}
\newtheorem{lemma}{Lemma}

\usepackage{setspace}
\usepackage[noadjust]{cite}
\makeatletter
\let\NAT@parse\undefined
\makeatother
\usepackage[pdftex, pdfborderstyle={/S/U/W 0}]{hyperref}

\usepackage{xcolor}
\usepackage{soul}
\definecolor{highlight}{rgb}{1,1,0.8}
\sethlcolor{highlight}
\definecolor{DarkOrange}{rgb}{0.8,0.3,0}

\newcommand{\bs}[1]{\boldsymbol{#1}}
\newcommand{\mat}[1]{\mathbf{#1}}
\newcommand{\pose}{\bs{\eta}}
\newcommand{\vel}{\bs{\nu}}
\newcommand{\ocean}{\mat{V}_c}
\newcommand{\oceanhat}{\hat{\mat{V}}_c}
\newcommand{\oceantilde}{\tilde{\mat{V}}_c}
\newcommand{\mass}{\mat{M}}
\newcommand{\cor}{\mat{C}\left(\bs{\nu}_r\right)}
\newcommand{\drag}{\mat{D}\left(\bs{\nu}_r\right)}
\newcommand{\grav}{\mat{g}\left(\pose\right)}
\newcommand{\rot}{\mat{R}\left(\theta, \psi\right)}
\newcommand{\trans}{\mat{T}\left(\theta\right)}
\newcommand{\transform}{\mat{J}\left(\pose\right)}
\newcommand{\T}{^{\rm T}}
\newcommand{\abs}[1]{\left|#1\right|}
\newcommand{\norm}[1]{\left\|#1\right\|}

\newcommand*{\scale}[2][1]{\scalebox{#1}{\ensuremath{#2}}}

\title{\LARGE \bf Formation Path Following Control of Underactuated AUVs -- With Proofs}

\author{Josef Matou\v{s}\quad Kristin Y. Pettersen\quad Claudio Paliotta
\thanks{Josef Matou\v{s} and Kristin Y. Pettersen are with the Centre for Autonomous Marine Operations and Systems, Department of Engineering Cybernetics, Norwegian University of Science and Technology (NTNU), Trondheim, Norway.
        {\fontfamily{qcr}\selectfont \{josef.matous, kristin.y.pettersen\}@ntnu.no}.
Claudio Paliotta is with SINTEF Digital, Trondheim, Norway
        {\fontfamily{qcr}\selectfont claudio.paliotta@sintef.no}%
}
}

\begin{document}
    \maketitle
    \thispagestyle{empty}
    \pagestyle{empty}

    \begin{abstract}
        This paper proposes a novel method for formation path following of multiple underactuated autonomous underwater vehicles.
        The method combines line-of-sight guidance with null-space-based behavioral control, allowing the vehicles to follow curved paths while maintaining the desired formation.
        We investigate the dynamics of the path-following error using cascaded systems theory, and show that the closed-loop system is uniformly semi-globally exponentially stable.
        We validate the theoretical results through numerical simulations.
    \end{abstract}

    \section{Introduction}
    Autonomous underwater vehicles (AUVs) are being increasingly used in a number of applications such as transportation, seafloor mapping, and the ocean energy industry.
    Some complex tasks need to be performed by a group of cooperating AUVs.
    Consequently, there is a need for control algorithms that can guide a formation of AUVs along a given path while avoiding collisions with each other.

    A comprehensive overview of various formation path-following methods is presented in \cite{das_cooperative_2016}.
    Most of these methods are based on two concepts: coordinated path-following \cite{borhaug_2006_formation,ghabcheloo_2006_coordinated}, and leader-follower \cite{rongxin_2010_leader,soorki_2011_robust}.
    In the \emph{coordinated path-following} approach, each vehicle follows a predefined path separately.
    Formation is then achieved by coordinating the motion of the vehicles along these paths.
    In this approach, the formation-keeping error (\emph{i.e.,} the difference between the actual and desired relative position of the vehicles) may initially grow as the vehicles converge to their predefined paths.
    In the \emph{leader-follower} approach, one leading vehicle follows the given path while the followers adjust their speed and position to obtain the desired formation shape.
    This approach tends to suffer from the lack of formation feedback due to unidirectional communication (\emph{i.e.,} the leader does not adjust its velocity based on the followers).

    The null-space-based behavioral (NSB) algorithm has also been proposed to solve the formation path-following problem \cite{arrichiello_formation_2006,antonelli_experiments_2009,pang_2019_formation,eek_formation_2020}.
    The NSB algorithm is a centralized control method that allows to combine several hierarchic tasks.
    In the NSB framework, it is possible to design the path-following, formation-keeping, and collision avoidance tasks independently.
    By combining these tasks, the vehicles exhibit the desired behavior.

    This paper aims to extend the results of \cite{eek_formation_2020}, where an NSB algorithm is used to guide two surface vessels moving in the horizontal plane.
    Specifically, we propose an algorithm that works with an arbitrary number of AUVs with five degrees of freedom (DOFs) moving in 3D.
    Similarly to \cite{eek_formation_2020}, we solve the path-following task using line-of-sight (LOS) guidance.
    Using the cascaded systems theory results of \cite{pettersen_lyapunov_2017}, we prove that the closed-loop system consisting of a 3D LOS guidance law, combined with surge, pitch, and yaw autopilots based on \cite{moe_LOS_2016}, is uniformly semi-globally exponentially stable (USGES) and uniformly globally asymptotically stable (UGAS).
    The theoretical results are verified through numerical simulations.

    The remainder of the paper is organized as follows.
    Section \ref{sec:model} introduces the model of the AUVs.
    In Section \ref{sec:objectives}, we define the formation path-following problem.
    In Section \ref{sec:control}, we describe the control system.
    The stability of the control system is proven in Section \ref{sec:path_stability}.
    Section \ref{sec:simulations} contains the results of a numerical simulation.
    Finally, Section \ref{sec:conclusion} contains some concluding remarks.

    \section{Model}
    \label{sec:model}
    In this section, we present the model of the AUV.
    We start by introducing the model in a matrix-vector form.
    Then, we write out the ordinary differential equations (ODEs) for the individual state variables.

    \subsection{Vehicle Model in Vector-Matrix Form}
    The pose ($\pose$) and velocities ($\vel$) of an AUV with 5DOFs are defined as
    \begin{align}
        \pose &= \left[x, y, z, \theta, \psi\right]\T, &
        \vel &= \left[u, v, w, q, r\right]\T,
    \end{align}
    where $x,y,z$ are the coordinates of the vehicle in North-East-Down (NED) coordinate frame, and $\theta$ and $\psi$ are the pitch and yaw angles, respectively.
    The velocities $u,v,w$ are the linear surge, sway and heave velocities in a given body-fixed frame, and $q$ and $r$ are the pitch and yaw rate, respectively.
    The roll dynamics are disregarded as the roll motion is assumed to be small and self-stabilizing by the vehicle design.

    Let $\ocean = \left[V_x, V_y, V_z\right]\T$ be the velocities of an unknown, constant and irrotational ocean current, given in the inertial NED frame.
    Let $\transform$ be the transformation matrix from the body-fixed to the inertial frame.
    $\transform$ is given by
    \begin{equation}
        \transform = \begin{bmatrix} \rot & \mat{O}_{3 \times 2} \\ \mat{O}_{2 \times 3} & \trans  \end{bmatrix},
    \end{equation}
    where $\rot$ is the rotation matrix from the body-fixed to the inertial coordinate frames, $\mat{O}_{n \times m}$ is an $n \times m$ matrix of zeros, and $\trans = {\rm diag} (1, 1/\cos\theta)$, which is well-defined if the pitch angle $\abs{\theta} < \pi/2$.
    Note that the mechanical design of torpedo-shaped rudder-controlled AUVs generally does not allow for pitch angles $\abs{\theta} = \pi/2$.

    The velocities of the ocean current expressed in the body-fixed coordinate frame, $\vel_c$, are thus \vspace{-2mm}
    \begin{equation}
        \vel_c = \left[ \left(\rot\T\,\ocean\right)\T, 0, 0 \right]\T. \label{eq:nu_c} \vspace{-2mm}
    \end{equation}
    We will denote the relative velocities of the vehicle as $\vel_r = \vel - \vel_c$.
    We will also denote the relative surge, sway and heave velocities as $u_r$, $v_r$ and $w_r$, respectively.

    Let $\mathbf{f} = \left[T_u, \delta_e, \delta_r\right]$ be the vector of control inputs, where $T_u$ is the surge thrust generated by the propeller, and $\delta_e$ and $\delta_r$ are the deflection angles of the elevator and rudder, respectively.
    Furthermore, let $\mass$ be the mass and inertia matrix, including added mass effects, $\cor$ the Coriolis centripetal matrix, also including added mass effects, and $\drag$ the hydrodynamic damping matrix.
    The dynamics of the vehicle in a matrix-vector form are then \cite{fossen_handbook_2011} \vspace{-2mm}
    \begin{align}
        \dot{\pose} &= \transform \vel, \label{eq:eta_dot}\\
        \mass \dot{\vel}_r + \cor\vel_r + \drag\vel_r + \grav &= \mat{B}\mat{f}, \label{eq:nu_dot} \vspace{-4.5mm}
    \end{align}
    where $\grav$ is the gravity and buoyancy vector, and $\mat{B}$ is the actuator configuration matrix that maps the control inputs to forces and torques.

    \vspace{-2mm}
    \subsection{Vehicle Model in Component Form} \vspace{-1mm}
    First, let us present the necessary assumptions for deriving the ODEs for individual state variables.
    \begin{assumption}
        \label{ass1}
        The vehicle is slender, torpedo-shaped with port-starboard symmetry.
    \end{assumption}
    \begin{assumption}
        \label{ass2}
        The hydrodynamic damping is linear.
    \end{assumption}
    \begin{assumption}
        \label{ass3}
        The vehicle is neutrally buoyant with the center of gravity (CG) and the center of buoyancy (CB) located along the same vertical axis.
    \end{assumption}

    \emph{Remark:} Assumptions \ref{ass1} and \ref{ass3} are valid from the mechanical design of commercial survey AUVs. Assumption \ref{ass2} is valid for low-speed missions. Also for higher-speed missions, this assumption is often made when designing the controller, as the higher-order damping coefficients are poorly known, and the forces are dissipative. Attempting to cancel the higher-order damping can thus introduce destabilizing control efforts.
    
    Under these assumptions, the $\mat{M}$ and $\mat{B}$ matrices have the following form
    \begin{align}
        \scale[0.9]{\mat{M}} & \scale[0.9]{= \begin{bmatrix} m_{11} & 0 & 0 & 0 & 0\\ 0 & m_{22} & 0 & 0 & m_{25}\\ 0 & 0 & m_{33} & m_{34} & 0\\ 0 & 0 & m_{34} & m_{44} & 0\\ 0 & m_{25} & 0 & 0 & m_{55} \end{bmatrix},} &
        \scale[0.9]{\mat{B}} & \scale[0.9]{= \begin{bmatrix}
            b_{11} & 0 & 0 \\
            0 & 0 & b_{23} \\
            0 & b_{32} & 0 \\
            0 & b_{42} & 0 \\
            0 & 0 & b_{53}
        \end{bmatrix}} \label{eq:mass}
    \end{align}
    the corresponding Coriolis matrix is
    \begin{equation}
        \cor = \begin{bmatrix} 0 & 0 & 0 & c_1 & -c_2 \\ 0 & 0 & 0 & 0 & c_3\\ 0 & 0 & 0 & -c_3 & 0\\ -c_1 & 0 & c_3 & 0 & 0\\ c_2 & -c_3 & 0 & 0 & 0 \end{bmatrix}, \label{eq:coriolis}
    \end{equation}
    where $c_1 = m_{34}\,q+m_{33}\,w_{r}$, $c_2 = m_{25}\,r+m_{22}\,v_{r}$, and $c_3 = m_{11}\,u_{r}$. The hydrodynamic damping matrix is
    \begin{equation}
        \drag \approx \mat{D} = \begin{bmatrix} d_{11} & 0 & 0 & 0 & 0\\ 0 & d_{22} & 0 & 0 & d_{25}\\ 0 & 0 & d_{33} & d_{34} & 0\\ 0 & 0 & d_{43} & d_{44} & 0\\ 0 & d_{52} & 0 & 0 & d_{55} \end{bmatrix}, \label{eq:drag}
    \end{equation}
    and the gravity vector has the following form
    \begin{equation}
        \grav = \left[0, 0, 0, m\,g\,z_g\,\sin(\theta)\right]\T,
    \end{equation}
    where $m$ is the weight of the vessel, $g$ is the gravity acceleration, and $z_g$ is the vertical distance between the CG and CB \cite{fossen_handbook_2011}.

    Furthermore, we assume that the actuators produce no sway and heave acceleration.
    In other words, for every $\mathbf{f}$ there exist $f_u, t_q$ and $t_r$ such that \vspace{-2mm}
    \begin{equation}
        \mass^{-1}\,\mat{B}\,\mat{f} = \left[f_u, 0, 0, t_q, t_r\right]\T.
        \label{eq:forces}\vspace{-2mm}
    \end{equation}
    In \cite{borhaug_straight_2007}, it is shown that if a vehicle satisfies Assumptions \ref{ass1}--\ref{ass3}, the origin of the body-fixed coordinate frame can always be chosen such that \eqref{eq:forces} holds.

    Under these assumptions, the model can expressed in the following form \vspace{-2mm}
    \begin{subequations}
        \begin{align}
            \scale[0.93]{\dot{x}} &= \scale[0.93]{u\,\cos\left(\psi \right)\,\cos\left(\theta \right)-v\,\sin\left(\psi \right)+w\,\cos\left(\psi \right)\,\sin\left(\theta \right),} \\
            \scale[0.93]{\dot{y}} &= \scale[0.93]{u\,\cos\left(\theta \right)\,\sin\left(\psi \right)+v\,\cos\left(\psi \right)+w\,\sin\left(\psi \right)\,\sin\left(\theta \right),} \\
            \scale[0.93]{\dot{z}} &= \scale[0.93]{-u\,\sin\left(\theta \right)+w\,\cos\left(\theta \right),} \\
            \scale[0.93]{\dot{\theta}} &= \scale[0.93]{q,} \label{eq:theta_dot} \\
            \scale[0.93]{\dot{\psi}} &= \scale[0.93]{\frac{1}{\cos\left(\theta \right)}\,r,} \label{eq:psi_dot} \\
            \scale[0.93]{\dot{u}} &= \scale[0.93]{f_u + F_u(u, v, w, q, r) + \bs{\phi}_u(u, v, w, q, r, \theta, \psi)\T\,\ocean,} \label{eq:u_dot} \\
            \scale[0.93]{\dot{v}} &= \scale[0.93]{X_v(u, u_c)\,r + Y_v(u, u_c)\,v_r,} \\
            \scale[0.93]{\dot{w}} &= \scale[0.93]{X_w(u, u_c)\,q + Y_w(u, u_c)\,w_r + G(\theta),} \\
            \scale[0.93]{\dot{q}} &= \scale[0.93]{t_q + F_q(u, w, q, \theta) + \bs{\phi}_q(u, w, q, \theta, \psi)\T\,\bs{\vartheta}\left(\ocean\right),} \label{eq:q_dot} \\
            \scale[0.93]{\dot{r}} &= \scale[0.93]{t_r + F_r(u, v, r) + \bs{\phi}_r(u, v, r, \theta, \psi)\T\,\bs{\vartheta}\left(\ocean\right),} \label{eq:r_dot} 
        \end{align} \label{eq:components} \vspace{-5mm}
    \end{subequations}

    \noindent where $\scale[0.93]{\bs{\vartheta}\left(\ocean\right) = \left[V_x, V_y, V_z, V_x^2, V_y^2, V_z^2, V_x\,V_y, V_x\,V_z, V_y\,V_z\right]\T}$, and the expressions for $F_i(\cdot)$, $\bs{\phi}_i(\cdot)$, $i \in \{u,q,r\}$, $X_v(\cdot)$, $Y_v(\cdot)$, $X_w(\cdot)$, $Y_w(\cdot)$, and $G(\cdot)$ are given in Appendix~\ref{app:components}.

    \section{Control Objectives}
    \label{sec:objectives}\vspace{-2mm}
    The goal is to control $n$ AUVs so that they move in a prescribed formation while avoiding collisions, and their barycenter follows a given path.

    The prescribed path in the inertial coordinate frame is given by a smooth function $\mat{p}_p: \mathbb{R} \rightarrow \mathbb{R}^3$.
    We assume that the path function is $\mathcal{C}^2$ and regular, \emph{i.e.,} the function is continuous up to its second derivative and its first derivative with respect to the path parameter satisfies\vspace{-1mm}
    \begin{equation}
        \norm{\frac{\partial \mat{p}_p(\xi)}{\partial \xi}} \neq 0.\vspace{-1mm}
    \end{equation}
    Therefore, for every point $\mat{p}_p(\xi)$ on the path, there exist path-tangential angles, $\theta_p(\xi)$ and $\psi_p(\xi)$, and a corresponding path-tangential coordinate frame $(x^p, y^p, z^p)$ (see Figure~\ref{fig:path}).
    
    The path-following error $\mat{p}_b^p$ is given by the position of the barycenter expressed in the path-tangential coordinate frame\vspace{-1mm}
    \begin{equation}
        \mat{p}_b^p = \mat{R}\left(\theta_p(\xi), \psi_p(\xi)\right)\T \, \big(\mat{p}_b - \mat{p}_p(\xi)\big), \label{eq:barycenter}\vspace{-2mm}
    \end{equation}
    where\vspace{-3mm}
    \begin{align}
        \mat{p}_b &= \frac{1}{n} \sum_{i=1}^n \mat{p}_i, & 
        \mat{p}_i &= \left[x_i, y_i, z_i\right]\T, \vspace{-3mm}
    \end{align}
    where $(x_i, y_i, z_i)$ is the position of the $i^{\rm th}$ vehicle.
    The goal of path following is to control the vehicles so that $\mat{p}_b^p \equiv \mat{0}_3$, where $\mat{0}_3$ is a 3-element vector of zeros.

    \begin{figure}[b]
        \centering
        \vspace*{-8mm}
        \def\svgwidth{.4\textwidth}
        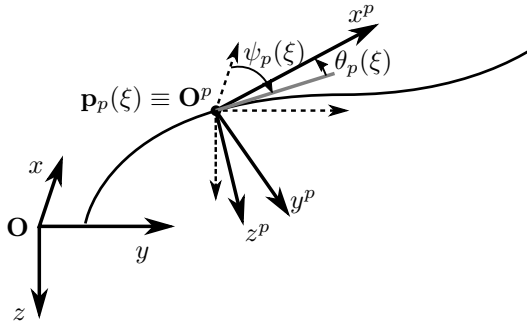
        \vspace{-3mm}
        \caption{Definition of the path angles and path-tangential coordinate frame. $\mathbf{O}$ denotes the origin of the inertial coordinate frame, $\mathbf{O}^p$ denotes the origin of the path-tangential frame, the grey line represents the projection of the path-tangential vector into the $xy$-plane.}
        \label{fig:path}
    \end{figure}

    \begin{figure}[t]
        \centering
        \vspace{3mm}
        \def\svgwidth{.33\textwidth}
        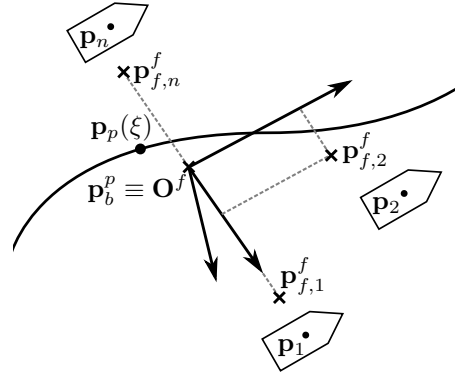
        \vspace{-1mm}
        \caption{Definition of the formation. $\mathbf{O}^f$ denotes the origin of the formation-centered coordinate frame.}
        \label{fig:formation}
        \vspace*{-7mm}
    \end{figure}

    To define the formation-keeping problem, we first define the formation-centered coordinate frame.
    This coordinate frame is created by translating the path-tangential frame into the barycenter (see Figure~\ref{fig:formation}).
    Let $\mat{p}_{f,1}^f, \ldots, \mat{p}_{f,n}^f$ be the position vectors that represent the desired formation.
    These vectors should be chosen such that their mean is equal to the barycenter.
    Since the barycenter is equivalent to the origin of the formation-centered frame, the vectors must thus satisfy\vspace{-1.5mm}
    \begin{equation}
        \sum_{i=1}^n \mat{p}_{f, i}^f = \mat{0}_3. \vspace{-0.5mm}
    \end{equation}

    \noindent The position of vehicle $i$ in the formation-centered frame is \vspace{-0.5mm}
    \begin{equation}
        \mat{p}_i^f = \mat{R}\T\left(\theta_p(\xi), \psi_p(\xi)\right) \left(\mat{p}_i - \mat{p}_b\right).
    \end{equation}    
    The goal of formation keeping is to have $\mat{p}_i^f \equiv \mat{p}_{f, i}^f$.
    This problem can also be expressed in the inertial coordinate frame as
    \begin{align}
        \mat{p}_i &\equiv \mat{R}\left(\theta_p(\xi), \psi_p(\xi)\right) \mat{p}_{f,i}^f + \mat{p}_b, &
        i &\in \left\{1, \ldots, n \right\}.
    \end{align}

    \section{Control System}
    \label{sec:control}
    To solve the formation path following problem, we propose a method that combines collision avoidance (COLAV), formation keeping, and LOS path following in a hierarchic manner using an NSB algorithm.
    Since the NSB algorithm outputs velocity references, we also need a low-level attitude control system to track these references.

    In this section, we first present the attitude control system.
    Then, in Section~\ref{sec:NSB}, we present the NSB algorithm and the associated COLAV and formation keeping tasks.
    In Section~\ref{sec:LOS}, we present the LOS guidance law for path following.
    Finally, in Section~\ref{sec:path_parameter}, we demonstrate how to use the path variable update law to cancel unwanted terms in the path following error dynamics.

    \subsection{Attitude Control System}
    \label{sec:ACS}
    This system controls the surge velocity, pitch, and yaw via the corresponding accelerations.
    The system is based on the autopilots in \cite{moe_LOS_2016}, but extended to 5DOFs.

    Let $u_d$ be the desired surge velocity and $\dot{u}_d$ its derivative.
    Let $\oceanhat$ be the estimate of the ocean current.
    Furthermore, let us define $\tilde{u} = u - u_d$ and $\oceantilde = \oceanhat - \ocean$.
    The surge controller consists of an output-linearizing sliding-mode P-controller and an ocean current observer \vspace{-1mm}
    \begin{align}
        f_u &= \dot{u}_d - F_u(\cdot) - \bs{\phi}_u(\cdot)\T\,\oceanhat - k_u\,\tilde{u} - k_c\,{\rm sign}\left(\tilde{u}\right), \label{eq:t_u} \\
        \dot{\hat{\mat{V}}}_c &= c_u\,\bs{\phi}_u(\cdot)\,\tilde{u}, \label{eq:V_hat_u} \vspace{-2mm}
    \end{align}
    where $k_u$, $k_c$ and $c_u$ are positive gains.

    Let $\theta_d$ be the desired pitch angle and $\dot{\theta}_d, \ddot{\theta}_d$ its derivatives.
    Let $\hat{\bs{\vartheta}}_q$ be the estimate of $\bs{\vartheta}(\ocean)$.
    Furthermore, let us define $\tilde{\theta} = \theta - \theta_d$, $\tilde{q} = q - \dot{\theta}_d$ and $\tilde{\bs{\vartheta}}_q = \hat{\bs{\vartheta}}_q - \bs{\vartheta}(\ocean)$.
    Inspired by \cite{moe_set-based_2017}, we introduce the following transformation \vspace{-1.5mm}
    \begin{equation}
        s_q = \tilde{q} + \lambda_q\,\tilde{\theta}, \vspace{-2mm}
    \end{equation}
    where $\lambda_q$ is a positive constant.
    The pitch controller consists of an output-linearizing sliding-mode PD-controller and an ocean current observer \vspace{-1mm}
    \begin{align}
        \begin{split}
            t_q &= \ddot{\theta}_d - F_q(\cdot) - \bs{\phi}_q(\cdot)\T\,\hat{\bs{\vartheta}}_q - \lambda_q\,\tilde{q} \\
            & \quad - k_{\theta}\,\tilde{\theta} - k_q\,s_q - k_d\,{\rm sign}(s_q), 
        \end{split} \label{eq:t_q} \\
        \dot{\hat{\bs{\vartheta}}}_q &= c_q\,\bs{\phi}_q(\cdot)\,s_q, \label{eq:V_hat_q}
    \end{align}
    \vspace{-5.5mm}

    \noindent where $k_{\theta}$, $k_q$, $k_d$ and $c_q$ are positive gains.

    Let $\psi_d$ be the desired yaw angle and $\dot{\psi}_d, \ddot{\psi}_d$ its derivatives.
    Let $\hat{\bs{\vartheta}}_r$ be the estimate of $\bs{\vartheta}(\ocean)$.
    Furthermore, let us define $\tilde{\psi} = \psi - \psi_d$ and $\tilde{\bs{\vartheta}}_r = \hat{\bs{\vartheta}}_r - \bs{\vartheta}(\ocean)$.
    Similarly to the pitch controller, we introduce the following transformation
    \begin{equation}
        s_r = \dot{\tilde{\psi}} + \lambda_r\,\tilde{\psi} = \frac{r}{\cos\theta} - \dot{\psi}_d + \lambda_r\,\tilde{\psi},
    \end{equation}
    where $\lambda_r$ is a positive constant.
    The yaw controller is analogous to the pitch controller introduced in the previous section
    \begin{align}        
        \begin{split}
            \scale[0.95]{t_r} &\scale[0.95]{= - F_r(\cdot) - \bs{\phi}_r(\cdot)\T\,\hat{\bs{\vartheta}}_r - r\,\tan(\theta)\dot{\theta}} \\
            & \scale[0.95]{ + \cos(\theta)\left(\ddot{\psi}_d - \lambda_r\,\dot{\tilde{\lambda}} - k_{\psi}\,\tilde{\psi} - k_r\,s_r - k_d\,{\rm sign}(s_r)\right), }
        \end{split} \label{eq:t_r} \\
        \scale[0.95]{\dot{\hat{\bs{\vartheta}}}_r} & \scale[0.95]{= c_r\,\bs{\phi}_r(\cdot)\,s_r,} \label{eq:V_hat_r}
    \end{align}
    where $k_{\psi}$, $k_r$, $k_d$ and $c_r$ are positive gains.

    \subsection{NSB Tasks}
    \label{sec:NSB}
    Let us denote the variables associated with the COLAV, formation keeping, and path following tasks by lower indices $1$, $2$, and $3$, respectively.
    Each task produces a vector of desired velocities, $\mat{v}_{d,i} \in \mathbb{R}^{3n},\, i\in\{1,2,3\}$.

    For the COLAV and formation keeping tasks, the desired velocities are obtained using task varibles, $\bs{\sigma}_1$ and $\bs{\sigma}_2$, and their desired values, $\bs{\sigma}_{d,1}$ and $\bs{\sigma}_{d,2}$.
    
    First, let us consider the COLAV task.
    Let {$d_{\rm COLAV}$} be the \emph{activation distance}, \emph{i.e.,} the distance at which the vehicles need to start performing the evasive maneuvers.
    The task variable is then given by a vector of relative distances between the vehicles smaller than $d_{\rm COLAV}$, \emph{i.e.,}
    \begin{align}
            \bs{\sigma}_1 &= \big[\norm{\mat{p}_i - \mat{p}_j}\big]\T, &
            \begin{split} 
                \forall &i,j\in\{1,\ldots,n\}, j > i, \\
                &\norm{\mat{p}_i - \mat{p}_j} < d_{\rm COLAV}.
            \end{split}
    \end{align}
    The desired values of the task are
    \begin{equation}
        \bs{\sigma}_{1,d} = d_{\rm COLAV} \, \mat{1},
    \end{equation}
    where $\mat{1}$ is a vector of ones.
    Note that this task does not guarantee robust collision avoidance.
    During the transients, the relative distance may become smaller than $d_{\rm COLAV}$.
    Therefore, to ensure collision avoidance, $d_{\rm COLAV}$ shuld be chosen as $d_{\rm min} + d_{\rm sec}$, where $d_{\rm min}$ is the minimum safe distance between the vehicles, and $d_{\rm sec}$ is an additional security distance.

    Now, let us consider the formation keeping task.
    The task variable is defined as
    \begin{align}
        \bs{\sigma}_2 &= \left[\bs{\sigma}_{2,1}\T, \ldots, \bs{\sigma}_{2,n-1}\T\right]\T, \label{eq:sigma_2} &
        \bs{\sigma}_{2,i} &= \mat{p}_i - \mat{p}_b,
    \end{align}
    and its desired values are
    \begin{equation}
        \bs{\sigma}_{d,2} = \begin{bmatrix}
            \mat{R}\left(\theta_p(\xi), \psi_p(\xi)\right)\,\mat{p}_{f,1}^p \\
            \vdots \\
            \mat{R}\left(\theta_p(\xi), \psi_p(\xi)\right)\,\mat{p}_{f,n-1}^p
        \end{bmatrix}. \label{eq:sigma_d_2}
    \end{equation}

    The desired velocities of the COLAV and formation keeping tasks are obtained using the closed-loop inverse kinematics (CLIK) equation \cite{arrichiello_formation_2006}
    \begin{align}
        \mat{v}_{d,i} &= \mat{J}_i^{\dagger}\,\left(\dot{\bs{\sigma}}_{d,i} - \bs{\Lambda}_i\,\tilde{\bs{\sigma}}_i\right), &
        i &\in \{1,2\}, \label{eq:CLIK}
    \end{align}
    where $\tilde{\bs{\sigma}}_i = \bs{\sigma}_i - \bs{\sigma}_{d,i}$ is the error, $\mat{J}^{\dagger}$ is the Moore-Penrose pseudoinverse, $\mat{\Lambda}_i$ is a positive definite gain matrix, and $\mat{J}_i$ is the task Jacobian given by
    \begin{align}
        \mat{J}_i &= \frac{\partial \bs{\sigma}_i}{\partial \mat{p}}, &
        \mat{p} &= \left[\mat{p}_1\T, \ldots, \mat{p}_n\T\right]\T.
    \end{align}

    The desired velocity of the path-following task is obtained using LOS guidance that is explained in the next section.
    These velocities are then combined using the NSB algorithm
    \begin{equation}
        \scale[0.95]{\mat{v}_{\rm NSB} = \mat{v}_{d, 1} + \left(\mat{I} - \mat{J}_1^{\dagger}\mat{J}_1\right)\left(\mat{v}_{d, 2} + \left(\mat{I} - \mat{J}_2^{\dagger}\mat{J}_2\right)\mat{v}_{d, 3}\right),}
    \end{equation}
    if there are active COLAV tasks, and
    \begin{equation}
        \mat{v}_{\rm NSB} = \mat{v}_{d, 2} + \left(\mat{I} - \mat{J}_2^{\dagger}\mat{J}_2\right)\mat{v}_{d, 3}, \label{eq:NSB_nominal}
    \end{equation}
    if there are none ($\mat{I}$ is an identity matrix).
    The NSB velocities must be decomposed into surge, pitch, and yaw references that can be tracked by the attitude control system presented in Section~\ref{sec:ACS}.
    Similarly to \cite{arrichiello_formation_2006}, we propose a method with angle of attack and sideslip compensation
    \begin{align}
        \scale[0.96]{u_{d, i}} & \scale[0.96]{= U_{{\rm NSB}, i}\,\frac{1 + \cos\left(\gamma_{{\rm NSB}, i} - \gamma_i\right)\cos\left(\chi_{{\rm NSB}, i} - \chi_i\right)}{2},} \label{eq:u_d}\\
        \scale[0.96]{\theta_{d, i}} & \scale[0.96]{= \gamma_{{\rm NSB}, i} + \alpha_{d, i}, \quad \alpha_{d, i} = \arctan\left(\frac{w_i}{u_{d, i}}\right),} \label{eq:theta_d} \\
        \scale[0.96]{\psi_{d, i}} & \scale[0.96]{= \chi_{{\rm NSB}, i} - \beta_{d, i}, \quad \beta_{d, i} = \arcsin\left(\frac{v_i}{\sqrt{u_{d, i}^2 + v_i^2 + w_i^2}}\right), } \label{eq:psi_d}
    \end{align}
    where $v_i$ and $w_i$ are the sway and heave velocities, and $\gamma_i$ and $\chi_i$ are the flight-path and course angles of the $i^{\rm th}$ vehicle, respectively, and $U_{{\rm NSB},i}$, $\gamma_{{\rm NSB},i}$ and $\chi_{{\rm NSB},i}$ are given by
    \begin{subequations}
        \begin{align}
            U_{{\rm NSB}, i} &= \norm{\mat{v}_{{\rm NSB}, i}}, \quad
            \mat{v}_{{\rm NSB}, i} = \begin{bmatrix} \dot{x}_{{\rm NSB}, i} \\ \dot{y}_{{\rm NSB}, i} \\ \dot{z}_{{\rm NSB}, i}\end{bmatrix}, \\*
            \gamma_{{\rm NSB}, i} &= - \arcsin\left(\frac{\dot{y}_{{\rm NSB}, i}}{U_{{\rm NSB}, i}}\right), \\
            \chi_{{\rm NSB}, i} &= \mathrm{arctan}_2 \left(\dot{y}_{{\rm NSB}, i}, \dot{x}_{{\rm NSB}, i}\right), 
        \end{align}
    \end{subequations}
    where ${\rm arctan}_2(y, x)$ is the four-quadrant inverse tan.

    \subsection{Line-of-Sight Guidance}
    \label{sec:LOS}
    The desired flight-path angle and course of the path-following task are given by the following LOS law
    \begin{align}
        \scale[1]{\gamma_{\rm LOS}} & \scale[1]{= \theta_p(\xi) + \arctan\left(\frac{z_b^p}{\Delta\left(\mat{p}_b^p\right)}\right),} \label{eq:gamma_LOS} \\
        \scale[1]{\chi_{\rm LOS}} & \scale[1]{= \psi_p(\xi) - \arctan\left(\frac{y_b^p}{\Delta\left(\mat{p}_b^p\right)}\right),}
    \end{align}
    where $\mat{p}_b^p = \left[x_b^p, y_b^p, z_b^p\right]\T$, and $\Delta\left(\mat{p}_b^p\right)$ is the lookahead distance.
    Inspired by \cite{belleter_2019_observer}, we choose the lookahead distance as
    \begin{equation}
        \Delta\left(\mat{p}_b^p\right) = \sqrt{\Delta_0^2 + \left(x_b^p\right)^2 + \left(y_b^p\right)^2 + \left(z_b^p\right)^2}, \label{eq:delta}
    \end{equation}
    where $\Delta_0 > 0$ is a constant.

    The desired velocity of the path-following task is then given by
    \begin{equation}
        \mat{v}_{d, 3} = \mat{1}_n \otimes \mat{v}_{\rm LOS},
    \end{equation}
    where $\cdot \otimes \cdot$ is the Kronecker tensor product, and
    \begin{equation}
        \mat{v}_{\rm LOS} = \begin{bmatrix}
            \cos\left(\chi_{\mathrm{LOS}}\right)\,\cos\left(\gamma_{\mathrm{LOS}}\right)\\ 
            \cos\left(\gamma_{\mathrm{LOS}}\right)\,\sin\left(\chi_{\mathrm{LOS}}\right)\\ 
            -\sin\left(\gamma_{\mathrm{LOS}}\right)
        \end{bmatrix} \, U_{\rm LOS},
    \end{equation}
    where $U_{\rm LOS}$ is the desired path-following speed.

    \subsection{Path Parametrization}
    \label{sec:path_parameter}
    Inspired by \cite{belleter_2019_observer}, we use the update law of the path variable $\xi$ to get desirable behavior of the along-track error ($x_b^p$).
    
    Note that the kinematics of the $i^{\rm th}$ vehicle can be alternatively expressed using the total speed ($U_i$) and the flight-path ($\gamma_i$) and course ($\chi_i$) angles of the vehicle as
    \begin{equation}
        \scale[0.91]{\dot{\mat{p}}_i = \left[\cos\left(\chi_i\right)\,\cos\left(\gamma_i\right) ,\, \cos\left(\gamma_i\right)\,\sin\left(\chi_i\right) ,\, -\sin\left(\gamma_i\right)\right]\T \, U_i.} \label{eq:kinematics}
    \end{equation}
    Now, let us investigate the kinematics of the barycenter.
    Differentiating \eqref{eq:barycenter} with respect to time and substituting \eqref{eq:kinematics} yields the following equations
    \begin{subequations}
        \begin{align}
            \begin{split}
                \dot{x}_b^p &= \frac{1}{n}\sum_{i=1}^n U_i\,\Omega_x\left(\gamma_i, \theta_p, \chi_i, \psi_p\right) \\
                &\quad - \scale{\norm{\frac{\partial \mat{p}_p(\xi)}{\partial \xi}}}\dot{\xi} + \omega_zy_b^p - \omega_yz_b^p,
            \end{split} \label{eq:x_pb} \\
            \dot{y}_b^p &= \frac{1}{n}\sum_{i=1}^n U_i\,\Omega_y\left(\gamma_i, \theta_p, \chi_i, \psi_p\right) + \omega_xz_b^p - \omega_zx_b^p, \label{eq:y_pb} \\
            \dot{z}_b^p &= \frac{1}{n}\sum_{i=1}^n U_i\,\Omega_z\left(\gamma_i, \theta_p, \chi_i, \psi_p\right) + \omega_yx_b^p - \omega_xy_b^p, \label{eq:z_pb}
        \end{align} \label{eq:barycenter_kinematics}
    \end{subequations}
    where
    \begin{subequations}
        \begin{align}
            \scale[0.90]{\Omega_x(\cdot)} & \mathrlap{\scale[0.90]{= \sin\left(\theta_p\right)\sin\left(\gamma_i\right) + \cos\left(\theta_p\right)\cos\left(\gamma_i\right)\cos\left(\psi_p-\chi_i\right),}} \\
            \scale[0.90]{\Omega_y(\cdot)} & \mathrlap{\scale[0.90]{= -\cos\left(\gamma_i\right)\sin\left(\psi_p - \chi_i\right),}} \\
            \scale[0.90]{\Omega_z(\cdot)} & \mathrlap{\scale[0.90]{=-\cos\left(\theta_p\right)\sin\left(\gamma_i\right) + \cos\left(\gamma_i\right)\sin(\theta_p)\cos\left(\psi_p-\chi_i\right)}} \\
            \scale[0.90]{\omega_x} & \scale[0.90]{= -\iota\dot{\xi}\sin(\theta_p),} &
            \scale[0.90]{\omega_y} & \scale[0.90]{= \kappa\dot{\xi},} &
            \scale[0.90]{\omega_z} & \scale[0.90]{= \iota\dot{\xi}\cos(\theta_p),} \\
            \scale[0.90]{\kappa(\xi)} & \scale[0.90]{= \frac{\partial \theta_p(\xi)}{\partial \xi},} &
            \scale[0.90]{\iota(\xi)} & \scale[0.90]{= \frac{\partial \psi_p(\xi)}{\partial \xi}.}
        \end{align}
    \end{subequations}
    To stabilize the along-track error dynamics, we choose the following path variable update 
    \begin{equation}
        \begin{split}
            \dot{\xi} &= \norm{\frac{\partial \mat{p}_p(\xi)}{\partial \xi}}^{-1} \left( \frac{1}{n}\sum_{i=1}^n U_i\,\Omega_x\left(\gamma_i, \theta_p, \chi_i, \psi_p\right)\right. \\
            & \qquad \qquad \qquad \qquad \left. + k_{\xi}\,\frac{x_b^p}{\sqrt{1+\left(x_b^p\right)^2}}\right),
        \end{split} \label{eq:path_update}
    \end{equation}
    where $k_{\xi} > 0$ is a constant.

    \section{Closed-Loop Analysis}
    \label{sec:path_stability}
    In this section, we investigate the closed-loop stability of the path following task.
    We define two error states, $\tilde{\mat{X}}_1$ and $\tilde{\mat{X}}_2$, as
    \begin{align}
        \tilde{\mat{X}}_1 &= \left[x_b^p, y_b^p, z_b^p\right]\T, 
        \tilde{\mat{X}}_2 = \left[\tilde{\mat{X}}_{2,1}\T, \ldots, \tilde{\mat{X}}_{2,n}\T\right]\T, \label{eq:error_states} \\
        \tilde{\mat{X}}_{2,i} &= \left[\tilde{u}_i, s_{q,i}, \tilde{\theta}_i, s_{r,i}, \tilde{\psi}_i\right]\T,
    \end{align}

    Now, we can take the barycenter kinematics from \eqref{eq:barycenter_kinematics} and express it in terms of $\tilde{\mat{X}}_1$ and $\tilde{\mat{X}}_2$ as
    \begin{subequations}
        \begin{align}
            &\dot{x}_b^p = -k_{\xi}\frac{x_b^p}{\sqrt{1+\left(x_b^p\right)^2}} + \omega_zy_b^p - \omega_yz_b^p, \label{eq:x_pb_CL} \\
            &\begin{aligned}
                \dot{y}_b^p &= - \frac{1}{n}\sum_{i=1}^n U_{d,i}{\frac{\cos\left(\gamma_{\rm LOS}\right)y_b^p}{\sqrt{\Delta\left(\mat{p}_b^p\right)^2 + \left(y_b^p\right)^2}}} + \omega_xz_b^p - \mathrlap{\omega_zx_b^p} \\
                & + G_y\big(\tilde{u}_1, \ldots, \tilde{u}_n, \tilde{\psi}_1, \ldots, \tilde{\psi}_n, \gamma_1, \ldots, \gamma_n, \\
                & \quad\quad \mathrlap{u_{d,1}, \ldots, u_{d,n}, v_1, \ldots, v_n, w_1, \ldots, w_n, \mat{p}_b^p, \psi_p\big),}
            \end{aligned} \\
            &\begin{aligned}[t]
                \dot{z}_b^p &= \frac{1}{n} \sum_{i=1}^n U_{d,i}\frac{z_b^p}{\sqrt{\Delta\left(\mat{p}_b^p\right)^2 + \left(z_b^p\right)^2}} + \omega_yx_b^p - \omega_xy_b^p \\
                & \,\mathrlap{+ G_z\big(\tilde{u}_1, \ldots, \tilde{u}_n, \tilde{\theta}_1, \ldots, \tilde{\theta}_n, \gamma_1, \ldots, \gamma_n, \chi_1, \ldots, \chi_n,} \\
                & \mathrlap{\quad\, u_{d,1}, \ldots, u_{d,n}, v_1, \ldots, v_n, w_1, \ldots, w_n, \mat{p}_b^p, \psi_p, \theta_p\big).}
            \end{aligned} \label{eq:z_pb_CL}
        \end{align} \label{eq:nominal}
    \end{subequations}

    \noindent The equations for $G_y(\cdot)$ and $G_z(\cdot)$ are given in Appendix~\ref{app:barycenter}.
    Substituting the attitude control system \eqref{eq:t_u}--\eqref{eq:V_hat_r} into vehicle dynamics \eqref{eq:components} yields the following closed-loop behavior of $\tilde{\mat{X}}_2$
    \begin{subequations}
        \begin{align}
            \scale[0.97]{\dot{\tilde{u}}_i} &= \scale[0.97]{- k_u\,\tilde{u}_i - k_c\,{\rm sign}\left(\tilde{u}_i\right) - \bs{\phi}_u(\cdot)\T\tilde{\mat{V}}_{c,i},} \label{eq:u_tilde} \\
            \scale[0.97]{\dot{s}_{q,i}} &= \scale[0.97]{-k_{\theta}\,\tilde{\theta}_i - k_q\,s_{q,i} - k_d\,{\rm sign}(s_{q,i}) - \bs{\phi}_q(\cdot)\T\,\tilde{\bs{\vartheta}}_{q,i}, }\\
            \scale[0.97]{\dot{\tilde{\theta}}_i} &= \scale[0.97]{s_{q,i} - \lambda_q\,\tilde{\theta}_i,} \\
            \scale[0.97]{\dot{s}_{r,i}} &= \scale[0.97]{-k_{\theta}\,\tilde{\theta}_i - k_r\,s_{r,i} - k_d\,{\rm sign}(s_{r,i}) - \bs{\phi}_r(\cdot)\T\,\tilde{\bs{\vartheta}}_{r,i},} \\
            \scale[0.97]{\dot{\tilde{\psi}}_i} &= \scale[0.97]{s_{r,i} - \lambda_r\,\tilde{\psi}_i,} \label{eq:psi_tilde}
        \end{align} \label{eq:perturbing}
    \end{subequations}
    the ocean current estimate errors
    \begin{subequations}
        \begin{align}
            \dot{\tilde{\mat{V}}}_{c,i} & = c_u\,\bs{\phi}_u(\cdot)\,\tilde{u}_i, \label{eq:V_c_tilde} \\
            \dot{\tilde{\bs{\vartheta}}}_{q,i} &= c_q\,\bs{\phi}_q(\cdot)\,s_{q,i}, \\
            \dot{\tilde{\bs{\vartheta}}}_{r,i} &= c_r\,\bs{\phi}_r(\cdot)\,s_{r,i}, \label{eq:theta_r_tilde}
        \end{align} \label{eq:estimates}
    \end{subequations}
    and the underactuated sway and heave dynamics
    \begin{align}
        \dot{v}_i &= X_v(u_i, u_c)\,r_i + Y_v(u_i, u_c)\,(v_i - v_c), \label{eq:v_dot} \\
        \dot{w}_i &= X_w(u_i, u_c)\,q_i + Y_w(u_i, u_c)\,(w_i - w_c) + G(\theta_i). \label{eq:w_dot}
    \end{align}

    To prove the stability of the closed-loop system, we need the results of the three following lemmas.
    The lemmas follow the same structure as the 2D case for two ASVs in \cite{eek_formation_2020}, and are extended to handle an arbitrary number of AUVs moving in 3D.
    \begin{lemma}
        \label{lemma_1}
        The trajectories of the closed-loop system \eqref{eq:nominal}--\eqref{eq:w_dot} are forward complete.
    \end{lemma}

    \begin{proof}
        The proof is given in Appendix~\ref{app:lemma_1}.
    \end{proof}

    \begin{lemma}
        \label{lemma_2}
        The underactuated sway and heave dynamics are bounded near the manifold $\left[\tilde{\mat{X}}_1\T, \tilde{\mat{X}}_2\T\right] = \mat{0}\T$ if $Y_v(u, u_c) < 0$, $Y_w(u, u_c) < 0$ and the curvature of the path satisfies
        \begin{align}
            \abs{\kappa(\xi)} &< \frac{n}{2}\abs{\frac{Y_w(u, u_c)}{X_w(u, u_c)}}, &
            \abs{\iota(\xi)} &< \frac{n}{2}\abs{\frac{Y_v(u, u_c)}{X_v(u, u_c)}}, \label{eq:curvature_limit}
        \end{align}
        for all $u > 0$ and $u_c \in [-\norm{\mat{V}_c}, \norm{\mat{V}_c}]$.
    \end{lemma}

    \begin{proof}
        The proof is given in Appendix~\ref{app:lemma_2}.
    \end{proof}

    \begin{lemma}
        \label{lemma_3}
        The underactuated sway and heave dynamics are bounded near the manifold $\tilde{\mat{X}}_2 = \mat{0}$, independently of $\tilde{\mat{X}}_1$ if the assumptions in Lemma~\ref{lemma_2} are satisfied and the constant term $\Delta_0$ in the lookahead distance \eqref{eq:delta} is chosen so that
        \begin{equation}
            \begin{split}
                \Delta_0 > \max&\left\{\frac{3}{n\abs{\frac{Y_v\left(u, u_c\right)}{X_v\left(u, u_c\right)}} - 2\abs{\iota(\xi)}}, \right. \\
                & \quad \left. \frac{3}{n\abs{\frac{Y_w\left(u, u_c\right)}{X_w\left(u, u_c\right)}} - 2\abs{\kappa(\xi)}}\right\},
            \end{split} \label{eq:lookahead_limit}
        \end{equation}
        for all $u > 0$ and $u_c \in [-\norm{\mat{V}_c}, \norm{\mat{V}_c}]$.
    \end{lemma}

    \begin{proof}
        The proof is given in Appendix~\ref{app:lemma_3}.
    \end{proof}

    \begin{theorem}
        The origin $\left[\tilde{\mat{X}}_1\T, \tilde{\mat{X}}_2\T\right] = \mat{0}\T$ of the system described by \eqref{eq:nominal},\eqref{eq:perturbing} is a USGES equilibrium point if the conditions of Lemmas~\ref{lemma_2} and \ref{lemma_3} hold and the maximum pitch angle of the path satisfies
        \begin{equation}
            \theta_{p, {\rm max}} = \max_{\xi \in \mathbb{R}} \abs{\theta_p(\xi)} < \frac{\pi}{4}.
            \label{eq:theta_max}
        \end{equation}
        Moreover, the ocean current estimate errors \eqref{eq:estimates} and the underactuated sway and heave dynamics \eqref{eq:v_dot}, \eqref{eq:w_dot} are bounded.
    \end{theorem}

    \emph{Remark:} Condition \eqref{eq:theta_max} is needed to ensure that $\abs{\gamma_{\rm LOS}} < \pi/2$.
    Indeed, from \eqref{eq:gamma_LOS}, the largest possible LOS reference angle is
    \begin{equation}
        \begin{split}
            \gamma_{\rm LOS, max} &= \theta_{p, {\rm max}} + \lim_{z_b^p \rightarrow \infty} \arctan\left(\scale[1]{\frac{z_b^p}{\sqrt{\Delta_0^2 + \left(z_b^p\right)^2}}}\right) \\
            &= \theta_{p, {\rm max}} + \frac{\pi}{4}.
        \end{split}
    \end{equation}
    With \eqref{eq:theta_max} satisfied, the cosine of $\gamma_{\rm LOS}$ is always positive. We will use this fact in the proof.
    
    \begin{proof}
    The proof follows along the lines of \cite{eek_formation_2020}, but is extended to an arbitrary number of 5DOF vehicles.
    We will also use the results of \cite{pettersen_lyapunov_2017} to prove that the system is USGES.

    In Lemmas~\ref{lemma_1}--\ref{lemma_3}, we have shown that the closed-loop system is forward complete and the underactuated sway and heave dynamics are bounded near the manifold $\tilde{\mat{X}}_2 = \mat{0}$.
    Since \eqref{eq:perturbing} is UGES \cite{moe_LOS_2016}, we can conclude that there exists a finite time $T > t_0$ such that the solutions of \eqref{eq:perturbing} will be sufficiently close to $\tilde{\mat{X}}_2 = \mat{0}$ to guarantee boundedness of $v_i$ and $w_i$.
    Having established that the underactuated dynamics are bounded, we will now utilize cascaded theory to analyze the cascade \eqref{eq:nominal}, \eqref{eq:perturbing}, where \eqref{eq:perturbing} perturbs the nominal dynamics \eqref{eq:nominal} through the terms $G_y(\cdot)$ and $G_z(\cdot)$.

    Now, consider the nominal dynamics of $\tilde{\mat{X}}_1$ (\emph{i.e.,} \eqref{eq:nominal} without the perturbing terms $G_y$ and $G_z$), and a Lyapunov function candidate
    \begin{equation}
        V(\tilde{\mat{X}}_1) = \frac{1}{2} \tilde{\mat{X}}_1\T\,\tilde{\mat{X}}_1 = \frac{1}{2} \left((x_b^p)^2 + (y_b^p)^2 + (z_b^p)^2\right), \label{eq:LCF}
    \end{equation}
    whose derivative along the trajectories of \eqref{eq:nominal} is
    \begin{subequations}
        \begin{align}
            \dot{V}(\tilde{\mat{X}}_1) &= - \tilde{\mat{X}}_1\T\,\mat{Q}\,\tilde{\mat{X}}_1, &
            \mat{Q} &= {\rm diag}(q_1, q_2, q_3), \\
            q_1 &= \scale[1]{\frac{k_{\xi}}{\sqrt{1+\left(x_b^p\right)^2}}}, &
            q_2 &= \scale[1]{{\frac{\frac{1}{n}\sum_{i=1}^n U_{d,i}\cos\left(\gamma_{\rm LOS}\right)}{\sqrt{\Delta\left(\mat{p}_b^p\right)^2 + \left(y_b^p\right)^2}}}}, \\
            &&q_3 &= \scale[1]{\frac{\frac{1}{n} \sum_{i=1}^n U_{d,i}}{\sqrt{\Delta\left(\mat{p}_b^p\right)^2 + \left(z_b^p\right)^2}}}. 
        \end{align}
    \end{subequations}
    Note that $\mat{Q}$ is positive definite, and the nominal system is thus UGAS.
    Furthermore, note that the following inequality
    \begin{subequations}
        \begin{align}
            \dot{V}(\tilde{\mat{X}}_1) &\leq - q_{\rm min} \norm{\tilde{\mat{X}}_1}^2, \\
            q_{\rm min} &= \scale[1]{\min \left\{ \frac{k_{\xi}}{\sqrt{1+r^2}}, \frac{\frac{1}{n} \sum_{i=1}^n U_{d,i}\cos\left(\gamma_{\rm LOS}\right)}{\sqrt{\Delta_0^2 + 4r^2}} \right\}},
        \end{align}
    \end{subequations}
    holds $\forall \tilde{\mat{X}}_1 \in \mathcal{B}_r$.
    Thus, the conditions of \cite[Theorem 5]{pettersen_lyapunov_2017} are fulfilled with $k_1 = k_2 = 1/2$, $a = 2$, and $k_3 = q_{\rm min}$, and the nominal system is USGES.

    As discussed in the proof of Lemma~\ref{lemma_1}, the perturbing system \eqref{eq:perturbing} is UGES, implying both UGAS and USGES.
    Furthermore, it is straightforward to show that the following holds for the Lyapunov function \eqref{eq:LCF}
    \begin{align}
        \left\| \frac{\partial V}{\partial \tilde{\mat{X}}_1} \right\| \, \left\|\tilde{\mat{X}}_1\right\| &= \left\|\tilde{\mat{X}}_1\right\|^2 = 2\,V\left(\tilde{\mat{X}}_1\right), \quad \forall \tilde{\mat{X}}_1, \\
        \left\| \frac{\partial V}{\partial \tilde{\mat{X}}_1} \right\| &= \left\|\tilde{\mat{X}}_1\right\| \leq \mu, \quad \quad \forall \left\|\tilde{\mat{X}}_1\right\| \leq \mu.
    \end{align}
    Therefore, \cite[Assumption 1]{pettersen_lyapunov_2017} is satisfied with $c_1 = 2$ and $c_2 = \mu$ for any $\mu > 0$.

    Finally, \cite[Assumption 2]{pettersen_lyapunov_2017} must be investigated.
    From \eqref{eq:G_y}, \eqref{eq:G_z}, it can be shown that for both perturbing terms there exist positive functions $\zeta_{y,1}(\cdot)$, $\zeta_{y,2}(\cdot)$, $\zeta_{z,1}(\cdot)$, $\zeta_{z,2}(\cdot)$, such that
    \begin{align}
        \left|G_y(\cdot)\right| &\leq \zeta_{y,1}\left(\norm{\tilde{\mat{X}}_2}\right) + \zeta_{y,2}\left(\norm{\tilde{\mat{X}}_2}\right)\norm{\tilde{\mat{X}}_1}, \\
        \left|G_z(\cdot)\right| &\leq \zeta_{z,1}\left(\norm{\tilde{\mat{X}}_2}\right) + \zeta_{z,2}\left(\norm{\tilde{\mat{X}}_2}\right)\norm{\tilde{\mat{X}}_1}.
    \end{align}
    Therefore, all conditions of \cite[Proposition 9]{pettersen_lyapunov_2017} are satisfied, and the closed-loop system is USGES.    
    \end{proof}

    \section{Simulation Results}
    \label{sec:simulations}
    \begin{figure*}[t]
        \centering
%
%
\definecolor{mycolor1}{rgb}{0.00000,0.44700,0.74100}%
\definecolor{mycolor2}{rgb}{0.85000,0.32500,0.09800}%
\definecolor{mycolor3}{rgb}{0.92900,0.69400,0.12500}%
\begin{tikzpicture}

\begin{axis}[%
width=65mm,
height=25.399mm,
at={(0mm,39mm)},
scale only axis,
xmin=0,
xmax=150,
xtick={0,50,100,150},
ymin=-5.5,
ymax=6.9,
ylabel style={font=\color{white!15!black},yshift=-2mm},
ylabel={Error [m]},
xlabel style={font=\color{white!15!black},yshift=2mm},
xlabel={$t$ [s]},
axis background/.style={fill=white},
axis x line*=bottom,
axis y line*=left,
title style={font=\bfseries,yshift=-0.5mm},
title={Path-following error},
legend columns=3
]
\addplot [color=mycolor1, line width=1.0pt]
  table[]{simout-1.tsv};
  \addlegendentry{$x_b^p$}

\addplot [color=mycolor2, line width=1.0pt]
  table[]{simout-2.tsv};
  \addlegendentry{$y_b^p$}

\addplot [color=mycolor3, line width=1.0pt]
  table[]{simout-3.tsv};
  \addlegendentry{$z_b^p$}  
\end{axis}

\begin{axis}[%
  width=65mm,
  height=25.699mm,
  at={(0mm,0mm)},
  scale only axis,
  xmin=0,
  xmax=150,
  xlabel style={font=\color{white!15!black}, yshift=1mm},
  xlabel={$t$ [s]},
  ymin=0,
  ymax=20.62,
  ylabel style={font=\color{white!15!black}},
  ylabel={Distance [m]},
  axis background/.style={fill=white},
  title style={font=\bfseries, yshift=-1mm},
  title={Distance between the vehicles},
  axis x line*=bottom,
  axis y line*=left,
  legend style={at={(0.97,0.03)}, anchor=south east, legend cell align=left, align=left, draw=white!15!black},
  legend columns=2
  ]
  \addplot [color=mycolor1, line width=1.0pt]
    table[]{simout-16.tsv};
  \addlegendentry{$d_{1,2}$~}
  
  \addplot [color=mycolor2, line width=1.0pt]
    table[]{simout-17.tsv};
  \addlegendentry{$d_{1,3}$}
  
  \addplot [color=mycolor3, line width=1.0pt]
    table[]{simout-18.tsv};
  \addlegendentry{$d_{2,3}$~}
  
  \addplot [color=black, dashed, line width=1.0pt]
    table[]{simout-19.tsv};
  \addlegendentry{$d_{\rm COLAV}$}
  
\end{axis}

\begin{axis}[%
width=65mm,
height=19.539mm,
at={(85mm,45.439mm)},
scale only axis,
xmin=0,
xmax=150,
xtick={0,50,100,150},
xticklabels={\empty},
ymin=-5.32018044501408,
ymax=2.72629468835089,
ylabel style={font=\color{white!15!black}},
ylabel={$\tilde{\sigma}_x$ [m]},
axis background/.style={fill=white},
title style={font=\bfseries,yshift=-0.5mm},
title={Formation-keeping error},
axis x line*=bottom,
axis y line*=left
]
\addplot[area legend, dashed, draw=black, fill=gray, fill opacity=0.5, forget plot]
table[] {simout-7.tsv}--cycle;

\addplot [color=mycolor1, line width=1.0pt, forget plot]
  table[]{simout-4.tsv};

\addplot [color=mycolor2, line width=1.0pt, forget plot]
  table[]{simout-5.tsv};

\addplot [color=mycolor3, line width=1.0pt, forget plot]
  table[]{simout-6.tsv};
\end{axis}

\begin{axis}[%
width=65mm,
height=19.539mm,
at={(85mm,22.72mm)},
scale only axis,
xmin=0,
xmax=150,
xtick={0,50,100,150},
xticklabels={\empty},
ymin=-20,
ymax=20,
ylabel style={font=\color{white!15!black}, xshift=2mm, yshift=-2mm},
ylabel={$\tilde{\sigma}_y$ [m]},
axis background/.style={fill=white},
axis x line*=bottom,
axis y line*=left,
legend style={at={(0.97,1.35)}, anchor=north east, legend cell align=left, align=left, draw=white!15!black}
]
\addplot[area legend, dashed, draw=black, fill=gray, fill opacity=0.5, forget plot]
table[] {simout-11.tsv}--cycle;

\addplot [color=mycolor1, line width=1.0pt]
  table[]{simout-8.tsv};
  \addlegendentry{Vehicle 1}

\addplot [color=mycolor2, line width=1.0pt]
  table[]{simout-9.tsv};
  \addlegendentry{Vehicle 2}

\addplot [color=mycolor3, line width=1.0pt]
  table[]{simout-10.tsv};
  \addlegendentry{Vehicle 3}

\end{axis}

\begin{axis}[%
width=65mm,
height=19.539mm,
at={(85mm,0mm)},
scale only axis,
xmin=0,
xmax=150,
xlabel style={font=\color{white!15!black},yshift=1mm},
xlabel={$t$ [s]},
ymin=-10,
ymax=20,
ylabel style={font=\color{white!15!black}, yshift=-2mm},
ylabel={$\tilde{\sigma}_z$ [m]},
axis background/.style={fill=white},
axis x line*=bottom,
axis y line*=left
]
\addplot[area legend, dashed, draw=black, fill=gray, fill opacity=0.5, forget plot]
table[] {simout-15.tsv}--cycle;

\addplot [color=mycolor1, line width=1.0pt, forget plot]
  table[]{simout-12.tsv};
\addplot [color=mycolor2, line width=1.0pt, forget plot]
  table[]{simout-13.tsv};
\addplot [color=mycolor3, line width=1.0pt, forget plot]
  table[]{simout-14.tsv};

\end{axis}

\end{tikzpicture}%
        \vspace{-1.5mm}
        \caption{Simulation results. The top-left plot shows the $x$-, $y$- and $z$-components of the path-following error $\mathbf{p}_b^p$, as defined in \eqref{eq:barycenter}. The bottom-left plot shows the distance between the vehicles ($d_{i,j} = \|\mathbf{p}_i - \mathbf{p}_j\|$). The plots on the right show the $x$-, $y$- and $z$-components of the formation-keeping error $\tilde{\bs{\sigma}} = \bs{\sigma}_2 - \bs{\sigma}_{d,2}$ with $\bs{\sigma}_2$ given by \eqref{eq:sigma_2} and $\bs{\sigma}_{d,2}$ given by \eqref{eq:sigma_d_2}. The grey rectangles mark the intervals when the COLAV task is active.}
        \label{fig:results}
        \vspace{-5mm}
    \end{figure*}
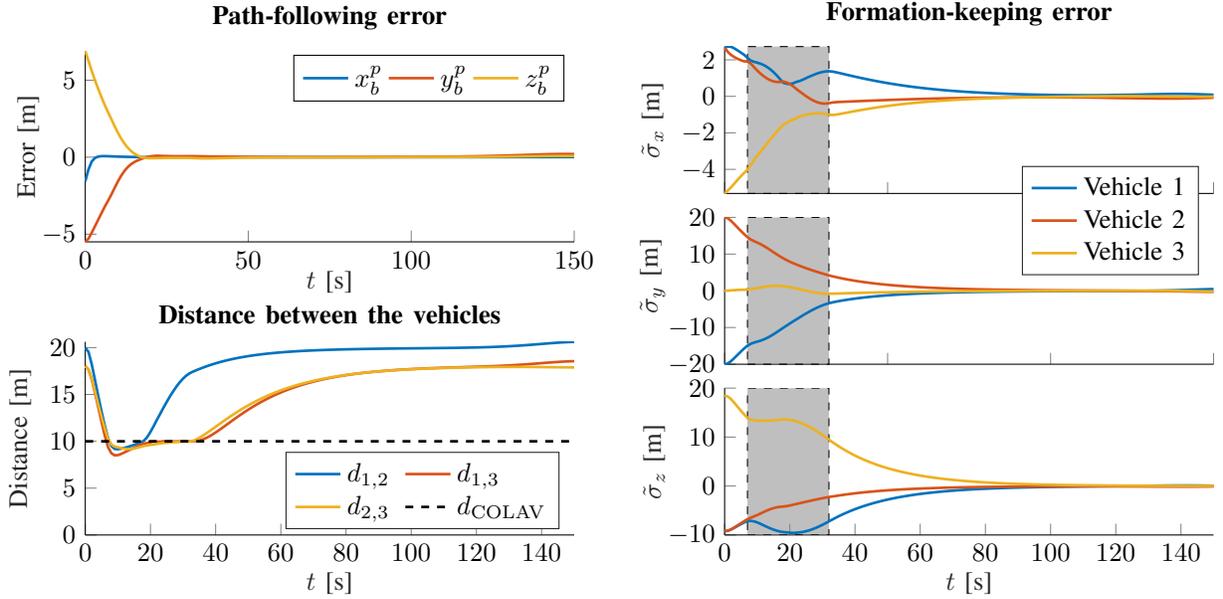

    In this section, we present the results of a numerical simulation of three LAUV vehicles \cite{sousa_LAUV_2012}.
    The parameters of the simulation are summarized in Table \ref{tab:params}.
    The barycenter should follow a spiral path given by
    \begin{equation}
        \mat{p}_p(\xi) = \left[\xi, a\,\cos(\omega\,\xi), b\,\sin(\omega\,\xi)\right]\T.
    \end{equation}
    The maximum curvature of this path is
    \begin{align}
        \max_{\xi \in \mathbb{R}} \abs{\kappa(\xi)} &= \frac{b\,\omega ^2}{\sqrt{a^2\,\omega ^2+1}}, &
        \max_{\xi \in \mathbb{R}} \abs{\iota(\xi)} &= a\,\omega ^2,
    \end{align}
    while the smallest absolute values of $Y_v / X_v$ and $Y_w / X_w$ for the LAUV model are approximately $0.26$.
    Consequently, the path is chosen such that the maximum curvature is
    \begin{align}
        \max_{\xi \in \mathbb{R}} \abs{\kappa(\xi)} &= 0.013, &
        \max_{\xi \in \mathbb{R}} \abs{\iota(\xi)} &= 0.040,
    \end{align}
    and \eqref{eq:curvature_limit} is satisfied.
    From \eqref{eq:lookahead_limit}, the lookahead distance must then satisfy $\Delta_0 > 4.29$.
    We choose $\Delta_0 = 5$, since smaller distances guarantee faster convergence.

    The very minimum relative distance to avoid collision is the length of the LAUV, \emph{i.e.} $2.4$ m.
    For additional safety, we design the COLAV task with $d_{\rm min} = 5$ m.
    To add a security zone during transients, $d_{\rm COLAV}$ is chosen to be $10$ m.

    The desired formation is an isosceles triangle parallel to the $yz$ plane.
    Specifically, the desired positions of the three vehicles are
    \begin{align}
        \mathbf{p}_{f,1}^f &= \begin{bmatrix} 0 \\ 10 \\ 5\end{bmatrix}, &
        \mathbf{p}_{f,2}^f &= \begin{bmatrix} 0 \\ -10 \\ 5\end{bmatrix}, &
        \mathbf{p}_{f,3}^f &= \begin{bmatrix} 0 \\ 0 \\ -10\end{bmatrix}.
    \end{align}

    \begin{table}[b]
        \centering
        \begin{tabular}[t]{r|l}
            {\bf Parameter} & {\bf Value} \\ \hline
            $k_u$ & $0.05$ \\
            $k_c$ & $0.1$ \\
            $k_{\theta}, k_{\psi}$ & $0.0625$ \\
            $k_q, k_r$ & $0.25$ \\
            $k_d$ & $0.1$ \\
            $\lambda_q, \lambda_r$ & $0.75$ \\
            $c_u$ & $5$ \\
            $c_q, c_r$ & $1$ \\
            $\bs{\Lambda}_1$ & $\mat{I}$ \\
            $\bs{\Lambda}_2$ & $0.05\,\mat{I}$
        \end{tabular}
        \begin{tabular}[t]{r|l}
            {\bf Parameter} & {\bf Value} \\ \hline
            $\mat{V}_c$ & $\left[0, 0.25, 0.05\right]\T$ \\
            $\Delta_0$ & $5$ \\
            $d_{\rm COLAV}$ & $10$ \\
            $U_{\rm LOS}$ & $1$ \\
            $k_{\xi}$ & $1$ \\
            $\mat{p}_0$ & $\mat{0}_3$ \\
            $a$ & $40$ \\
            $b$ & $20$ \\
            $\omega$ & $\pi / 100$
        \end{tabular}
        \caption{Simulation parameters}
        \label{tab:params}
    \end{table}

    The gains of the low-level control systems \eqref{eq:t_u},\eqref{eq:t_q},\eqref{eq:t_r} are chosen such that the settling time is approximately 10 seconds.
    The gains of the pitch and yaw PD controllers are chosen such that the closed-loop system is critically damped.

    The results of the numerical simulation are shown in Figures \ref{fig:results} and \ref{fig:trajectory}.
    The vehicles start in an inverted triangular formation.
    The COLAV task is briefly activated, and the distance between the vehicles drops to approximately 8 meters during the transient.
    Eventually, the vehicles resolve the situation and continue to converge to the desired path and formation.
    
    Note that while the COLAV task is active, the formation-keeping error is diverging.
    After resolving the situation, the formation-keeping error converges to zero exponentially.
    The rate of convergence is given by the formation-keeping gain $\bs{\Lambda}_2$.

    The path-following error seems to converge linearly at first, and then exponentially as the error gets smaller.
    This phenomenon is caused by the LOS guidance law \eqref{eq:gamma_LOS}, \emph{cf.} \cite{fossen_uniform_2014}, and the path parameter update law \eqref{eq:path_update}.
    The inverse tan in \eqref{eq:gamma_LOS} and the last term in \eqref{eq:path_update} act as a saturation, slowing the convergence for large errors.
    The rate of convergence of the along-track error ($x_b^p$) is given by the path parameter update gain $k_{\xi}$, while the rate of convergence of the cross-track errors ($y_b^p, z_b^p$) is given by the lookahead distance $\Delta_0$.

    \begin{figure}[t]
        \centering
        \def\svgwidth{.48\textwidth}
        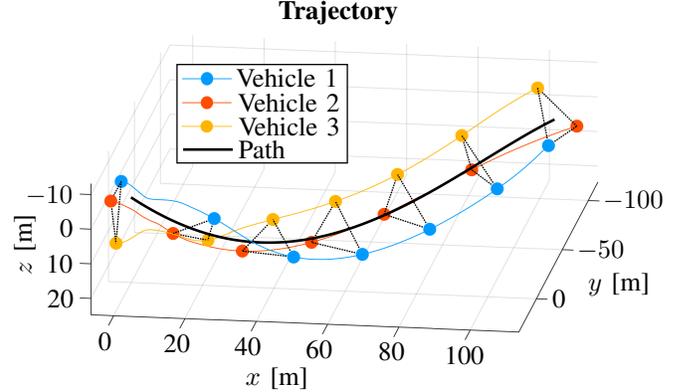
        \vspace{-6.5mm}
        \caption{3D trajectory of the vehicles. The markers represent the position of the vehicles at times $t = 0, 25, 50, \ldots, 150$ seconds. Markers with corresponding times are connected by dotted lines to better illustrate the resulting formation.}
        \label{fig:trajectory}
        \vspace*{-8mm}
    \end{figure}

    \section{Conclusions and future work}
    \label{sec:conclusion}


    In this paper, we proposed a formation path-following method for an arbitrary number of AUVs, proved the stability of the path following part, and verified its effectiveness in simulations.

    Because the proposed algorithm is centralized, our method can only be used in scenarios where all the vehicles can communicate with each other.
    A decentralized version of the algorithm is a topic for future work.

    In the simulations, the formation-keeping error shows exponential convergence to zero.
    However, the stability of the formation-keeping task has not been theoretically proven.
    Proving the stability of this task is another potential topic for future work.

    \section*{ACKNOWLEDGMENTS}
    The authors would like to thank Damiano Varagnolo for the interesting discussions and inputs.

    This work was partly supported by the Research Council of Norway through project No. 302435 and the Centres of Excellence funding scheme, project No. 223254.


    \bibliographystyle{IEEEtran}
    \bibliography{biblio}

\begin{thebibliography}{10}
\providecommand{\url}[1]{#1}
\csname url@samestyle\endcsname
\providecommand{\newblock}{\relax}
\providecommand{\bibinfo}[2]{#2}
\providecommand{\BIBentrySTDinterwordspacing}{\spaceskip=0pt\relax}
\providecommand{\BIBentryALTinterwordstretchfactor}{4}
\providecommand{\BIBentryALTinterwordspacing}{\spaceskip=\fontdimen2\font plus
\BIBentryALTinterwordstretchfactor\fontdimen3\font minus
  \fontdimen4\font\relax}
\providecommand{\BIBforeignlanguage}[2]{{%
\expandafter\ifx\csname l@#1\endcsname\relax
\typeout{** WARNING: IEEEtran.bst: No hyphenation pattern has been}%
\typeout{** loaded for the language `#1'. Using the pattern for}%
\typeout{** the default language instead.}%
\else
\language=\csname l@#1\endcsname
\fi
#2}}
\providecommand{\BIBdecl}{\relax}
\BIBdecl

\bibitem{das_cooperative_2016}
B.~Das, B.~Subudhi, and B.~B. Pati, ``\BIBforeignlanguage{en}{Cooperative
  formation control of autonomous underwater vehicles: {An} overview},''
  \emph{\BIBforeignlanguage{en}{International Journal of Automation and
  Computing}}, vol.~13, no.~3, pp. 199--225, Jun. 2016.

\bibitem{borhaug_2006_formation}
E.~Borhaug and K.~Y. Pettersen, ``Formation control of 6-{DOF}
  {E}uler-{L}agrange systems with restricted inter-vehicle communication,'' in
  \emph{Proc. 45th IEEE Conf. Decision and Control}, 2006, pp. 5718--5723.

\bibitem{ghabcheloo_2006_coordinated}
R.~Ghabcheloo, A.~P. Aguiar, A.~Pascoal, C.~Silvestre, I.~Kaminer, and
  J.~Hespanha, ``Coordinated path-following control of multiple underactuated
  autonomous vehicles in the presence of communication failures,'' in
  \emph{Proc. 45th IEEE Conference on Decision and Control}, 2006, pp.
  4345--4350.

\bibitem{rongxin_2010_leader}
R.~Cui, S.~{Sam Ge}, B.~{Voon Ee How}, and Y.~{Sang Choo}, ``Leader–follower
  formation control of underactuated autonomous underwater vehicles,''
  \emph{Ocean Engineering}, vol.~37, no.~17, pp. 1491--1502, 2010.

\bibitem{soorki_2011_robust}
M.~Soorki, H.~Talebi, and S.~Nikravesh, ``A robust dynamic leader-follower
  formation control with active obstacle avoidance,'' in \emph{Proc. 2011 IEEE
  International Conference on Systems, Man, and Cybernetics}, 2011, pp.
  1932--1937.

\bibitem{arrichiello_formation_2006}
F.~Arrichiello, S.~Chiaverini, and T.~I. Fossen, ``Formation control of
  underactuated surface vessels using the null-space-based behavioral
  control,'' in \emph{Proc. 2006 IEEE/RSJ International Conference on
  Intelligent Robots and Systems}, 2006, pp. 5942--5947.

\bibitem{antonelli_experiments_2009}
G.~Antonelli, F.~Arrichiello, and S.~Chiaverini, ``Experiments of formation
  control with multirobot systems using the null-space-based behavioral
  control,'' \emph{IEEE Transactions on Control Systems Technology}, vol.~17,
  no.~5, pp. 1173--1182, 2009.

\bibitem{pang_2019_formation}
S.-K. Pang, Y.-H. Li, and H.~Yi, ``Joint formation control with obstacle
  avoidance of towfish and multiple autonomous underwater vehicles based on
  graph theory and the null-space-based method,'' \emph{Sensors}, vol.~19,
  no.~11, 2019.

\bibitem{eek_formation_2020}
{\r{A}}.~Eek, K.~Y. Pettersen, E.-L.~M. Ruud, and T.~R. Krogstad, ``Formation
  {Path} {Following} {Control} of {Underactuated} {USVs},'' \emph{European
  Journal of Control}, Jun. 2021.

\bibitem{pettersen_lyapunov_2017}
K.~Y. Pettersen, ``\BIBforeignlanguage{en}{Lyapunov sufficient conditions for
  uniform semiglobal exponential stability},''
  \emph{\BIBforeignlanguage{en}{Automatica}}, vol.~78, pp. 97--102, Apr. 2017.

\bibitem{moe_LOS_2016}
S.~Moe, K.~Y. Pettersen, T.~I. Fossen, and J.~T. Gravdahl, ``Line-of-sight
  curved path following for underactuated {USV}s and {AUV}s in the horizontal
  plane under the influence of ocean currents,'' in \emph{Proc. 24th
  Mediterranean Conf. Control and Automation}, 2016, pp. 38--45.

\bibitem{fossen_handbook_2011}
T.~I. Fossen, \emph{\BIBforeignlanguage{en}{Handbook of {Marine} {Craft}
  {Hydrodynamics} and {Motion} {Control}}}.\hskip 1em plus 0.5em minus
  0.4em\relax John Wiley \& Sons, May 2011.

\bibitem{borhaug_straight_2007}
E.~Borhaug, A.~Pavlov, and K.~Y. Pettersen, ``Straight line path following for
  formations of underactuated underwater vehicles,'' in \emph{Proc. 46th IEEE
  Conf. Decision and Control}, 2007, pp. 2905--2912.

\bibitem{moe_set-based_2017}
S.~Moe and K.~Y. Pettersen, ``Set-based line-of-sight ({LOS}) path following
  with collision avoidance for underactuated unmanned surface vessels under the
  influence of ocean currents,'' in \emph{Proc. 2017 IEEE Conf. Control
  Technology and Applications}, 2017, pp. 241--248.

\bibitem{belleter_2019_observer}
D.~Belleter, M.~A. Maghenem, C.~Paliotta, and K.~Y. Pettersen, ``Observer based
  path following for underactuated marine vessels in the presence of ocean
  currents: A global approach,'' \emph{Automatica}, vol. 100, pp. 123--134,
  2019.

\bibitem{sousa_LAUV_2012}
A.~Sousa, L.~Madureira, J.~Coelho, J.~Pinto, J.~Pereira, J.~{Borges Sousa}, and
  P.~Dias, ``{LAUV}: {T}he man-portable autonomous underwater vehicle,''
  \emph{IFAC Proceedings Volumes}, vol.~45, no.~5, pp. 268--274, 2012.

\bibitem{fossen_uniform_2014}
T.~I. Fossen and K.~Y. Pettersen, ``On uniform semiglobal exponential stability
  (usges) of proportional line-of-sight guidance laws,'' \emph{Automatica},
  vol.~50, no.~11, pp. 2912--2917, 2014.

\bibitem{angeli_forward_1999}
D.~Angeli and E.~D. Sontag, ``Forward completeness, unboundedness
  observability, and their {Lyapunov} characterizations,'' \emph{Systems \&
  Control Letters}, vol.~38, no. 4-5, pp. 209--217, 1999.

\end{thebibliography}

    \clearpage
    \onecolumn
    \appendices
    \section{Components of the dynamical equations}
    \label{app:components}
    \begin{subequations}
        \begin{align}
            F_u(\cdot) &= -\frac{d_{11}\,u+q\,\left(m_{34}\,q+m_{33}\,w\right)-r\,\left(m_{25}\,r+m_{22}\,v\right)}{m_{11}}, \\
            \bs{\phi}_u(\cdot) &= q\left(\frac{m_{33}}{m_{11}} - 1\right)\,\mat{r}_w + r\left(1 - \frac{m_{22}}{m_{11}}\right)\,\mat{r}_v + \frac{d_{11}}{m_{11}}\,\mat{r}_u, \\
            X_v(\cdot) &= - u_c - \frac{m_{55}\,\left(d_{25}+m_{11}\,u_{r}\right)-m_{25}\,\left(d_{55}+m_{25}\,u_{r}\right)}{m_{22}\,m_{55}-{m_{25}}^2}, \\
            Y_v(\cdot) &= - \frac{d_{22}\,m_{55}-m_{25}\,\left(d_{52}-u_{r}\,\left(m_{11}-m_{22}\right)\right)}{m_{22}\,m_{55}-{m_{25}}^2}, \\
            X_w(\cdot) &= u_c - \frac{m_{44}\,\left(d_{34}-m_{11}\,u_{r}\right)-m_{34}\,\left(d_{44}-m_{34}\,u_{r}\right)}{m_{33}\,m_{44}-{m_{34}}^2}, \\
            Y_w(\cdot) &= - \frac{d_{33}\,m_{44}-m_{34}\,\left(d_{43}+u_{r}\,\left(m_{11}-m_{33}\right)\right)}{m_{33}\,m_{44}-{m_{34}}^2}, \\
            G(\cdot) &= \frac{m_{34}\,mgz_g\,\sin\left(\theta \right)}{m_{33}\,m_{44}-{m_{34}}^2}, \\
            F_q(\cdot) &= \frac{m_{34}\,\left(d_{34}\,q+d_{33}\,w-q\,u\,\left(m_{11}-m_{33}\right)\right) - m_{33}\,\left(d_{44}\,q+d_{43}\,w+mgz_g\,\sin\left(\theta \right)+u\,w\,\left(m_{11}-m_{33}\right)\right)}{m_{33}\,m_{44}-{m_{34}}^2}, \\
            \bs{\phi}_q(\cdot) &= \frac{m_{34}\,\left(d_{33}\,\bs{\varphi}_{w}-\bs{\varphi}_{u}\,q\,\left(m_{11}-m_{33}\right)\right)-m_{33}\,\left(d_{43}\,\bs{\varphi}_{w}+\bs{\varphi}_{{uw}}\,\left(m_{11}-m_{33}\right)\right)}{m_{33}\,m_{44}-{m_{34}}^2}, \\
            F_r(\cdot) &= \frac{m_{25}\,\left(d_{25}\,r+d_{22}\,v+r\,u\,\left(m_{11}-m_{22}\right)\right)-m_{22}\,\left(d_{55}\,r+d_{52}\,v-u\,v\,\left(m_{11}-m_{22}\right)\right)}{m_{22}\,m_{55}-{m_{25}}^2},\\
            \bs{\phi}_r(\cdot) &= \frac{m_{25}\,\left(d_{22}\,\bs{\varphi}_{v}+\bs{\varphi}_{u}\,r\,\left(m_{11}-m_{22}\right)\right)-m_{22}\,\left(d_{52}\,\bs{\varphi}_{v}-\bs{\varphi}_{uv}\,\left(m_{11}-m_{22}\right)\right)}{m_{22}\,m_{55}-{m_{25}}^2},
        \end{align}
    \end{subequations}
    where $\left[\mat{r}_u, \mat{r}_v, \mat{r}_w\right] = \rot$, and
    \begin{equation}
        \bs{\varphi}_{ij}(\cdot) = \left[-j\,\mat{r}_i\T - i\,\mat{r}_j\T , r_{{i1}}\,r_{{j1}} , r_{{i2}}\,r_{{j2}} , r_{{i3}}\,r_{{j3}} , r_{{i1}}\,r_{{j2}}+r_{{i2}}\,r_{{j1}} , r_{{i1}}\,r_{{j3}}+r_{{i3}}\,r_{{j1}} , r_{{i2}}\,r_{{j3}}+r_{{i3}}\,r_{{j2}} \right]\T, \quad i, j \in \{u,v,w\}.
    \end{equation}

    \section{Derivations and Lemmas from Section~\ref{sec:path_stability}}
    \subsection{Derivation of Closed-Loop Barycenter Kinematics}
    \label{app:barycenter}
    We begin by taking $\dot{y}_b^p$ from \eqref{eq:y_pb}.
    \begin{align}
        \dot{y}_b^p &= \frac{1}{n}\sum_{i=1}^n U_i\,\cos\left(\gamma_i\right)\,\sin\left(\chi_i - \psi_p\right) - \dot{\xi}\,\iota\,x_b^p. \label{eq:y_pb_0}
    \end{align}
    Now, consider the term $\sin\left(\chi_i - \psi_p\right)$.
    The course of the vessel is given by
    \begin{align}
        \chi_i &= \psi_i + \beta_i, &
        \beta_i &= \arcsin\left(\frac{v_i}{U_i}\right).
    \end{align}
    After substituting and applying some trigonometric identities, we get
    \begin{subequations}
        \begin{align}
            \sin\left(\chi_i - \psi_p\right) &= \sin\left(\psi_i + \beta_i - \psi_p\right) = \cos\left(\psi_i - \psi_p\right)\,\sin\left(\beta_i\right) + \sin\left(\psi_i - \psi_p\right)\,\cos\left(\beta_i\right) \\
            &= \cos\left(\psi_i - \psi_p\right)\frac{v_i}{U_i} + \sin\left(\psi_i - \psi_p\right)\frac{\sqrt{u_i^2 + w_i^2}}{U_i}.
        \end{align}
    \end{subequations}
    Consequently, the term $U_i\,\cos\left(\gamma_i\right)\,\sin\left(\chi_i - \psi_p\right)$ is equivalent to
    \begin{equation}
        U_i\,\cos\left(\gamma_i\right)\,\sin\left(\chi_i - \psi_p\right) = \cos\left(\gamma_i\right) \left(\cos\left(\psi_i - \psi_p\right)v_i + \sin\left(\psi_i - \psi_p\right)\sqrt{u_i^2 + w_i^2}\right). \label{eq:y_pb_1}
    \end{equation}

    \noindent Now, consider a term $\sin\left(\psi_i + \beta_{d,i} - \psi_p\right)$.
    Using a similar procedure, we get
    \begin{equation}
        \sin\left(\psi_i + \beta_{d,i} - \psi_p\right) = \cos\left(\psi_i - \psi_p\right)\frac{v_i}{U_{d,i}} + \sin\left(\psi_i - \psi_p\right)\frac{\sqrt{u_{d,i}^2 + w_i^2}}{U_{d,i}}. \label{eq:y_pb_2}
    \end{equation}
    Combining \eqref{eq:y_pb_1} and \eqref{eq:y_pb_2}, we get
    \begin{equation}
        U_i\,\cos\left(\gamma_i\right)\,\sin\left(\chi_i - \psi_p\right) = U_{d,i}\,\cos\left(\gamma_i\right)\,\sin\left(\psi_i + \beta_{d,i} - \psi_p\right) + \cos\left(\gamma_i\right)\,\sin\left(\psi_i - \psi_p\right)\left(\sqrt{u_i^2 + w_i^2} - \sqrt{u_{d,i}^2 + w_i^2}\right). \label{eq:y_pb_3}
    \end{equation}
    Note that the following holds for the angles
    \begin{align}
        \psi_i + \beta_{d,i} - \psi_p &= \psi_{d,i} + \tilde{\psi}_i + \beta_{d,i} - \left(\psi_{d,i} + \beta_{d,i} + \beta_{\rm LOS}\right) = \tilde{\psi}_i - \beta_{\rm LOS}, &
        \beta_{\rm LOS} &= \scale[1]{\arctan\left(\frac{y_b^p}{\Delta\left(\mat{p}_b^p\right)}\right)}.
    \end{align}
    Therefore, their sine is given by
    \begin{equation}
        \sin\left(\psi_i + \beta_{d,i} - \psi_p\right) = \sin\left(\tilde{\psi}_i\right)\,\scale[1]{\frac{\Delta\left(\mat{p}_b^p\right)}{\sqrt{\Delta\left(\mat{p}_b^p\right)^2 + \left(y_b^p\right)^2}}} - \cos\left(\tilde{\psi}_i\right)\scale[1]{\frac{y_b^p}{\sqrt{\Delta\left(\mat{p}_b^p\right)^2 + \left(y_b^p\right)^2}}}.
        \label{eq:sin_psi_beta}
    \end{equation}
    Furthermore, note that the following holds for the flight-path angle
    \begin{equation}
        \gamma_i = \theta_i - \alpha_i = \tilde{\theta}_i + \theta_{d,i} - \alpha_i = \tilde{\theta}_i + \gamma_{\rm LOS} + \alpha_{d,i} - \alpha_i.
    \end{equation}
    Consequently, the cosine of the flight-path angle is equal to
    \begin{equation}
        \begin{split}
        \cos\left(\gamma_i\right) &= \cos\left(\gamma_{\rm LOS}\right)\cos\left(\tilde{\theta}_i\right)\cos\left(\alpha_{d,i} - \alpha_i\right) 
            - \cos\left(\gamma_{\rm LOS}\right)\sin\left(\tilde{\theta}_i\right)\sin\left(\alpha_{d,i} - \alpha_i\right) \\
            & \quad - \sin\left(\gamma_{\rm LOS}\right)\cos\left(\tilde{\theta}_i\right)\sin\left(\alpha_{d,i} - \alpha_i\right)
            - \sin\left(\gamma_{\rm LOS}\right)\sin\left(\tilde{\theta}_i\right)\cos\left(\alpha_{d,i} - \alpha_i\right)
        \end{split}
        \label{eq:cos_gamma_i}
    \end{equation}
    Using the equalities \eqref{eq:sin_psi_beta}, \eqref{eq:cos_gamma_i}, we can rewrite \eqref{eq:y_pb_3} as
    \begin{equation}
        U_i\,\cos\left(\gamma_i\right)\,\sin\left(\chi_i - \psi_p\right) = - U_{d,i}\cos\left(\gamma_{\rm LOS}\right)\scale[1]{\frac{y_b^p}{\sqrt{\Delta\left(\mat{p}_b^p\right)^2 + \left(y_b^p\right)^2}}} + G_{y,i}\left(\tilde{u}_i, \tilde{\psi}_i, \gamma_i, u_{d,i}, v_i, w_i, \mat{p}_b^p, \psi_p\right), \label{eq:y_pb_4}
    \end{equation}
    where
    \begin{equation}
        \begin{split}
            G_{y,i}(\cdot) &= \cos\left(\gamma_i\right)\,\sin\left(\psi_i - \psi_p\right)\left(\sqrt{u_i^2 + w_i^2} - \sqrt{u_{d,i}^2 + w_i^2}\right)
            - U_{d,i}\cos\left(\gamma_i\right)\,\sin\left(\tilde{\psi}_i\right)\scale[1]{\frac{\Delta\left(\mat{p}_b^p\right)}{\sqrt{\Delta\left(\mat{p}_b^p\right)^2 + \left(y_b^p\right)^2}}} \\
            &\quad + U_{d,i}\bigg[\sin\left(\gamma_{\rm LOS}\right)\left(\cos\left(\tilde{\theta}_i\right)\sin\left(\alpha_{d,i} - \alpha_i\right) + \sin\left(\tilde{\theta}_i\right)\cos\left(\alpha_{d,i} - \alpha_i\right)\right) \\
            & \qquad \qquad -\cos\left(\gamma_{\rm LOS}\right)\left(\cos\left(\tilde{\theta}_i\right)\cos\left(\alpha_{d,i} - \alpha_i\right) - 1\right)\bigg]\scale[1]{\frac{y_b^p}{\sqrt{\Delta\left(\mat{p}_b^p\right)^2 + \left(y_b^p\right)^2}}} 
        \end{split} \label{eq:G_y}
    \end{equation}

    \noindent Substituting \eqref{eq:y_pb_4} into \eqref{eq:y_pb_0}, we get the following
    \begin{equation}
        \begin{split}
            \dot{y}_b^p &= - \frac{1}{n}\sum_{i=1}^n U_{d,i}\cos\left(\gamma_{\rm LOS}\right)\scale[1]{\frac{y_b^p}{\sqrt{\Delta\left(\mat{p}_b^p\right)^2 + \left(y_b^p\right)^2}}} - \dot{\xi}\,\iota\,x_b^p \\
            &\quad + {G_y\left(\tilde{u}_1, \ldots, \tilde{u}_n, \tilde{\psi}_1, \ldots, \tilde{\psi}_n, \gamma_1, \ldots, \gamma_n, u_{d,1}, \ldots, u_{d,n}, v_1, \ldots, v_n, w_1, \ldots, w_n, \mat{p}_b^p, \psi_p\right)},
        \end{split}
    \end{equation}
    where
    \begin{equation}
        G_y(\cdot) = \frac{1}{n} \sum_{i=1}^n G_{y,i}\left(\tilde{u}_i, \tilde{\psi}_i, \gamma_i, u_{d,i}, v_i, w_i, \mat{p}_b^p, \psi_p\right).
    \end{equation}

    Now, we demonstrate a similar procedure for $\dot{z}_b^p$.
    From \eqref{eq:y_pb}, we get
    \begin{subequations}
        \begin{align}
            \dot{z}_b^p &= \frac{1}{n} \sum_{i=1}^n U_i \left(-\cos\left(\theta_p\right)\sin\left(\gamma_i\right) + \cos\left(\gamma_i\right)\sin(\theta_p)\cos\left(\psi_p-\chi_i\right)\right) + \dot{\xi}\,\kappa\,x_b^p \\
            &= \frac{1}{n} \sum_{i=1}^n U_i \left(-\sin\left(\gamma_i - \theta_p\right) - \left(1 - \cos\left(\chi_i - \psi_p\right)\right)\cos\left(\gamma_i\right)\sin(\theta_p)\right) + \dot{\xi}\,\kappa\,x_b^p. \label{eq:z_pb_1}
        \end{align} 
    \end{subequations}
    Once again, we consider the terms
    \begin{equation}
        \sin\left(\gamma_i - \theta_p\right) = \sin\left(\theta_i - \alpha_i - \theta_p\right) = \sin\left(\theta_i - \theta_p\right)\frac{u_i}{U_i} - \cos\left(\theta_i - \theta_p\right)\frac{w_i}{U_i},
    \end{equation}
    and
    \begin{equation}
        \sin\left(\theta_i - \alpha_{d,i} - \theta_p\right) = \sin\left(\theta_i - \theta_p\right)\frac{u_{d,i}}{U_{d,i}} - \cos\left(\theta_i - \theta_p\right)\frac{w_i}{U_{d,i}},
    \end{equation}
    which give us the following equality
    \begin{equation}
        U_i\,\sin\left(\gamma_i - \theta_p\right) = U_{d,i}\,\sin\left(\theta_i - \alpha_{d,i} - \theta_p\right) + \tilde{u}_i\,\sin\left(\theta_i - \theta_p\right).
    \end{equation}
    Using a similar trick, we can write the sine as
    \begin{equation}
        \sin\left(\theta_i - \alpha_{d,i} - \theta_p\right) = \sin\left(\tilde{\theta}_i - \alpha_{\rm LOS}\right) = \sin\left(\tilde{\theta}_i\right)\frac{\Delta\left(\mat{p}_b^p\right)}{\sqrt{\Delta\left(\mat{p}_b^p\right)^2 + \left(z_b^p\right)^2}} - \cos\left(\tilde{\theta}_i\right)\frac{\left(z_b^p\right)}{\sqrt{\Delta\left(\mat{p}_b^p\right)^2 + \left(z_b^p\right)^2}}
    \end{equation}
    Consequently, we can rewrite \eqref{eq:z_pb_1} as
    \begin{equation}
        \begin{split}
            \dot{z}_b^p &= - \frac{1}{n} \sum_{i=1}^n U_{d,i}\frac{z_b^p}{\sqrt{\Delta\left(\mat{p}_b^p\right)^2 + \left(z_b^p\right)^2}} + \dot{\xi}\,\kappa\,x_b^p \\
            & \quad + {G_z\left(\tilde{u}_1, \ldots, \tilde{u}_n, \tilde{\theta}_1, \ldots, \tilde{\theta}_n, \gamma_1, \ldots, \gamma_n, \chi_1, \ldots, \chi_n, u_{d,1}, \ldots, u_{d,n}, v_1, \ldots, v_n, w_1, \ldots, w_n, \mat{p}_b^p, \theta_p, \psi_p\right)},
        \end{split}
    \end{equation}
    where
    \begin{align}
        G_z(\cdot) &= \frac{1}{n} \sum_{i=1}^n G_{z,i}\left(\tilde{u}_i, \tilde{\theta}_i, \gamma_i, \chi_i, u_{d,i}, v_i, w_i, \mat{p}_b^p, \theta_p, \psi_p\right), \\
        \begin{split}
            G_{z,i}(\cdot) &= -U_i\left(\left(1 - \cos\left(\chi_i - \psi_p\right)\right)\cos\left(\gamma_i\right)\sin(\theta_p)\right) - \tilde{u}_i\,\sin\left(\theta_i - \theta_p\right) \\
            & \quad - \left(1 - \cos\left(\tilde{\theta}_i\right)\right)\frac{\left(z_b^p\right)}{\sqrt{\Delta\left(\mat{p}_b^p\right)^2 + \left(z_b^p\right)^2}} - U_{d,i}\sin\left(\tilde{\theta}_i\right)\frac{\Delta\left(\mat{p}_b^p\right)}{\sqrt{\Delta\left(\mat{p}_b^p\right)^2 + \left(z_b^p\right)^2}}.
        \end{split} \label{eq:G_z}
    \end{align}

    \subsection{Desired Pitch and Yaw Rate}
    For further calculations, we need to evaluate the desired pitch ($q_{d,i}$) and yaw ($r_{d,i}$) rates of the vessels.
    From \eqref{eq:theta_dot}, we get the following relation between the yaw rate and the derivative of the yaw angle
    \begin{equation}
        q_{d,i} = \dot{\theta}_{d,i}.
    \end{equation}
    Now, we consider the desired pitch angle from \eqref{eq:theta_d}.
    Since we are investigating the path following task, we substitute $\gamma_{\rm LOS}$ from \eqref{eq:gamma_LOS} for $\gamma_{{\rm NSB}, i}$.
    Differentiating \eqref{eq:theta_d} with respect to time yields
    \begin{subequations}
        \begin{align}
            q_{d,i} &= \dot{\theta}_p(\xi) + \frac{\Delta\left(\mat{p}_b^p\right)\,\dot{z}_b^p - z_b^p\,\dot{\Delta}\left(\mat{p}_b^p\right)}{\Delta\left(\mat{p}_b^p\right)^2 + \left(z_b^p\right)^2} + \frac{u_{d,i}\,\dot{w}}{u_{d,i}^2 + w_i^2} \\
            \begin{split}
                &= \dot{\xi}\,\kappa(\xi) + \frac{\Delta\left(\mat{p}_b^p\right)\left(\frac{1}{n} \sum_{j=1}^n U_{d,j}\frac{\left(z_b^p\right)}{\sqrt{\Delta\left(\mat{p}_b^p\right)^2 + \left(z_b^p\right)^2}} + \dot{\xi}\,\kappa\,x_b^p + G_z(\cdot)\right)}{\Delta\left(\mat{p}_b^p\right)^2 + \left(z_b^p\right)^2} \\
                &\quad + \frac{z_b^p\left(-k_{\xi}\frac{\left(x_b^p\right)^2}{\sqrt{1+\left(x_b^p\right)^2}} - \frac{1}{n}\sum_{j=1}^n U_{d,j}\left(\scale[1]{\frac{\cos\left(\gamma_{{\rm LOS},j}\right)^2\left(y_b^p\right)^2}{\sqrt{\Delta\left(\mat{p}_b^p\right)^2 + \left(y_b^p\right)^2}}} + \frac{\left(z_b^p\right)^2}{\sqrt{\Delta\left(\mat{p}_b^p\right)^2 + \left(z_b^p\right)^2}}\right) + y_b^p\,G_y(\cdot) + z_b^p\,G_z(\cdot)\right)}{\Delta\left(\mat{p}_b^p\right)\left(\Delta\left(\mat{p}_b^p\right)^2 + \left(z_b^p\right)^2\right)} \\
                &\quad +u_{d,i}\frac{X_w\left(u_{d,i}+\tilde{u}_i, u_c\right)\,q + Y_w\left(u_{d,i}+\tilde{u}_i, u_c\right)\left(w_i - w_c\right)}{u_{d,i}^2 + w_i^2}.
            \end{split}
        \end{align} \label{eq:q_d}
    \end{subequations}

    From \eqref{eq:psi_dot}, we get the following relation between the yaw rate and the derivative of the yaw angle
    \begin{equation}
        r_{d,i} = \dot{\psi}_{d,i}\,\cos\left(\theta_{d,i}\right).
    \end{equation}
    Substituting the time-derivative of \eqref{eq:psi_d}, we get
    \begin{subequations}
        \begin{align}
            r_{d,i} &= \left(\dot{\psi}_p(\xi) - \frac{\Delta\left(\mat{p}_b^p\right)\,\dot{y}_b^p - y_b^p\,\dot{\Delta}\left(\mat{p}_b^p\right)}{\Delta\left(\mat{p}_b^p\right)^2 + \left(y_b^p\right)^2} - \frac{\dot{v}}{\sqrt{U_{d,i}^2 - v_i^2}}\right)\,\cos\left(\theta_{d,i}\right) \\
            \begin{split}
                &= \left(\dot{\xi}\,\iota(\xi) - \frac{\Delta\left(\mat{p}_b^p\right)\left(\frac{1}{n} \sum_{j=1}^n U_{d,i}\frac{\cos\left(\gamma_{\rm LOS}\right)\left(y_b^p\right)}{\sqrt{\Delta\left(\mat{p}_b^p\right)^2 + \left(y_b^p\right)^2}} - \dot{\xi}\,\iota\,x_b^p + G_y(\cdot)\right)}{\Delta\left(\mat{p}_b^p\right)^2 + \left(y_b^p\right)^2} \right. \\
                &\qquad + \frac{y_b^p\left(-k_{\xi}\frac{\left(x_b^p\right)^2}{\sqrt{1+\left(x_b^p\right)^2}} - \frac{1}{n}\sum_{j=1}^n U_{d,i}\left(\scale[1]{\frac{\cos\left(\gamma_{\rm LOS}\right)^2\left(y_b^p\right)^2}{\sqrt{\Delta\left(\mat{p}_b^p\right)^2 + \left(y_b^p\right)^2}}} + \frac{\left(z_b^p\right)^2}{\sqrt{\Delta\left(\mat{p}_b^p\right)^2 + \left(z_b^p\right)^2}}\right) + y_b^p\,G_y(\cdot) + z_b^p\,G_z(\cdot)\right)}{\Delta\left(\mat{p}_b^p\right)\left(\Delta\left(\mat{p}_b^p\right)^2 + \left(y_b^p\right)^2\right)} \\
                &\qquad -\frac{X\left(u_{d,i}+\tilde{u}_i, u_c\right)\,r + Y\left(u_{d,i}+\tilde{u}_i, u_c\right)\left(v_i - v_c\right)}{\sqrt{u_{d,i}^2 + w_i^2}} \Bigg)\, \cos\left(\theta_{d,i}\right).
            \end{split}
        \end{align} \label{eq:r_d}
    \end{subequations}

    \subsection{Proof of Lemma \ref{lemma_1}}
    \label{app:lemma_1}
    In \cite{moe_LOS_2016}, it is shown that the error states \eqref{eq:u_tilde}--\eqref{eq:psi_tilde} are UGES and the ocean current estimate errors \eqref{eq:V_c_tilde}--\eqref{eq:theta_r_tilde} are bounded, which implies that \eqref{eq:u_tilde}--\eqref{eq:theta_r_tilde} are forward complete.
    Therefore, we only need to prove that the underactuated sway and heave dynamics \eqref{eq:v_dot}, \eqref{eq:w_dot} and the barycenter dynamics \eqref{eq:x_pb_CL}--\eqref{eq:z_pb_CL} are forward complete.

    First, let us consider the underactuated sway dynamics.
    From \eqref{eq:v_dot}, we get
    \begin{equation}
        \dot{v}_i = X_v\left(\tilde{u}_i + u_{d,i}, u_c\right)\,\left(\tilde{r}_i + r_{d,i}\right) + Y_v\left(\tilde{u}_i + u_{d,i}, u_c\right)\,\left(v_i - v_c\right),
    \end{equation}
    where $\tilde{r}_i = r_i - r_{d,i}$.
    Now, let us consider a Lyapunov function candidate
    \begin{equation}
        V_v(v_i) = \frac{1}{2} v_i^2. \label{eq:V_v}
    \end{equation}
    Its derivative along the trajectories of $v_i$ is
    \begin{equation}
        \dot{V}_v(v_i) = X_v\left(\tilde{u}_i + u_{d,i}, u_c\right)\,\left(\tilde{r}_i + r_{d,i}\right)\,v_i + Y_v\left(\tilde{u}_i + u_{d,i}, u_c\right)\,\left(v_i - v_c\right)\,v_i. \label{eq:V_v_dot}
    \end{equation}
    From the boudedness of $\tilde{\mat{X}}_{2,i}$, $\kappa(\xi)$, $\iota(\xi)$, $u_{d,i}$, $u_c$ and $v_c$, we can conclude that there exists some scalar $\beta_{v,0} > 0$ such that $\left\| \left[ \tilde{\mat{X}}_{2,i}\T, \kappa(\xi), \iota(\xi), u_{d,i}, u_c, v_c \right]\T \right\| \leq \beta_0$.
    Moreover, from \eqref{eq:r_d}, we can conclude that there exist some positive functions $a_r(\beta_{v,0})$ and $b_r(\beta_{v,0})$ such that
    \begin{equation}
        \abs{r_{d,i}} \leq a_r(\beta_{v,0})\,\abs{v_i} + b_r(\beta_{v,0}).
    \end{equation}
    Consequently, we can upper bound $\dot{V}_v(v_i)$ using the following expression
    \begin{equation}
        \dot{V}_v(v_i) \leq X_v\left(\tilde{u}_i + u_{d,i}, u_c\right)\left(\tilde{r}_i\,v_i + a_r(\cdot)v_i^2 + b_r(\cdot)v_i\right) + Y_v\left(\tilde{u}_i + u_{d,i}, u_c\right)\left(v_i^2 - v_c\,v_i\right).
    \end{equation}
    Using Young's inequality, we get
    \begin{subequations}
        \begin{align}
            \begin{split}
                \dot{V}_v(v_i) &\leq \left(X_v\left(\tilde{u}_i + u_{d,i}, u_c\right)\left(2 + a_r(\cdot)\right) + 2\,Y_v\left(\tilde{u}_i + u_{d,i}, u_c\right)\right)\,v_i^2 \\
                 & \quad + X_v\left(\tilde{u}_i + u_{d,i}, u_c\right)\left(\tilde{r}_i^2 + b_r(\cdot)^2\right) + Y_v\left(\tilde{u}_i + u_{d,i}, u_c\right)\,v_c^2
            \end{split} \\
            & \leq \alpha_v\,V_v(v_i) + \beta_v.
        \end{align}
    \end{subequations}
    Using the comparison lemma, we get
    \begin{equation}
        V_v\left(v_i(t)\right) \leq \left(V_v\left(v_i(t_0)\right) + \frac{\beta_v}{\alpha_v}\right)\,{\rm exp}\left(\alpha_v(t - t_0)\right) - \frac{\beta_v}{\alpha_v}.
    \end{equation}
    As $V_v(v_i)$ is defined for all $t > t_0$, it follows that $v_i$ is also defined for all $t > t_0$.
    The solutions of \eqref{eq:v_dot} thus fulfill the definition of forward completeness, as defined in \cite{angeli_forward_1999}.

    Now, let us consider the underactuated heave dynamics.
    From \eqref{eq:w_dot}, we get
    \begin{equation}
        \dot{w}_i = X_w\left(\tilde{u}_i + u_{d,i}, u_c\right)\,\left(\tilde{q}_i + q_{d,i}\right) + Y_w\left(\tilde{u}_i + u_{d,i}, u_c\right)\,\left(w_i - w_c\right) + G(\theta_i),
    \end{equation}
    where $\tilde{q}_i = q_i - q_{d,i}$.
    Similar to the previous paragraph, we consider a Lyapunov function candidate
    \begin{equation}
        V_w(w_i) = \frac{1}{2} w_i^2, \label{eq:V_w}
    \end{equation}
    whose derivative is
    \begin{equation}
        \dot{V}_w(w_i) = X_w\left(\tilde{u}_i + u_{d,i}, u_c\right)\,\left(\tilde{q}_i + q_{d,i}\right)\,w_i + Y_w\left(\tilde{u}_i + u_{d,i}, u_c\right)\,\left(w_i - w_c\right)\,w_i + G(\theta)\,w_i.
    \end{equation}
    From the boudedness of $\tilde{\mat{X}}_{2,i}$, $\kappa(\xi)$, $\iota(\xi)$, $u_{d,i}$, $u_c$ and $w_c$, we can conclude that there exists some scalar $\beta_0 > 0$ such that $\left\| \left[ \tilde{\mat{X}}_{2,i}\T, \kappa(\xi), \iota(\xi), u_{d,i}, u_c, w_c \right]\T \right\| \leq \beta_{w,0}$.
    Moreover, from \eqref{eq:q_d}, we can conclude that there exist some positive functions $a_q(\beta_{w,0})$ and $b_q(\beta_{w,0})$ such that
    \begin{equation}
        \abs{q_{d,i}} \leq a_q(\beta_{w,0})\,\abs{w_i} + b_q(\beta_{w,0}).
    \end{equation}
    Consequently, we can upper bound $\dot{V}_w(w_i)$ using the following expression
    \begin{equation}
        \dot{V}_w(w_i) \leq X_w\left(\tilde{u}_i + u_{d,i}, u_c\right)\left(\tilde{q}_i\,w_i + a_q(\cdot)w_i^2 + b_q(\cdot)w_i\right) + Y_w\left(\tilde{u}_i + u_{d,i}, u_c\right)\left(w_i^2 - w_c\,w_i\right) + G(\theta_i)\,w_i.
    \end{equation}
    Using Young's inequality, we get
    \begin{subequations}
        \begin{align}
            \begin{split}
                \dot{V}_w(w_i) &\leq \left(X_w\left(\tilde{u}_i + u_{d,i}, u_c\right)\left(2 + a_q(\cdot)\right) + 2\,Y_w\left(\tilde{u}_i + u_{d,i}, u_c\right) + 1\right)\,w_i^2 \\
                 & \quad + X_w\left(\tilde{u}_i + u_{d,i}, u_c\right)\left(\tilde{q}_i^2 + b_q(\cdot)^2\right) + Y_w\left(\tilde{u}_i + u_{d,i}, u_c\right)\,w_c^2 + G(\theta)^2
            \end{split} \\
            & \leq \alpha_w\,V_w(w_i) + \beta_w.
        \end{align}
    \end{subequations}
    Using the comparison lemma, we get
    \begin{equation}
        V_w\left(w_i(t)\right) \leq \left(V_w\left(w_i(t_0)\right) + \frac{\beta_w}{\alpha_w}\right)\,{\rm exp}\left(\alpha_w(t - t_0)\right) - \frac{\beta_w}{\alpha_w}.
    \end{equation}
    Using the same arguments as in the previous paragraph, we conclude that the solutions of \eqref{eq:w_dot} are forward complete.

    Finally, let us consider the barycenter dynamics.
    We use a Lyapunov function candidate
    \begin{equation}
        V_b(\mat{p}_b^p) = \frac{1}{2} \left(\left(x_b^p\right)^2 + \left(y_b^p\right)^2 + \left(z_b^p\right)^2\right),
    \end{equation}
    whose derivative along the solutions of \eqref{eq:x_pb_CL}--\eqref{eq:z_pb_CL} is
    \begin{subequations}
        \begin{align}
            \dot{V}_b\left(\mat{p}_b^p\right) &= -k_{\xi}\frac{\left(x_b^p\right)^2}{\sqrt{1 + \left(x_b^p\right)^2}} - \frac{1}{n}\sum_{i=1}^n U_{d,i} \left(
                \frac{\cos\left(\gamma_{\rm LOS}\right)^2\left(y_b^p\right)^2}{\sqrt{\Delta\left(\mat{p}_b^p\right)^2 + \left(y_b^p\right)^2}} +
                \frac{\left(z_b^p\right)^2}{\sqrt{\Delta\left(\mat{p}_b^p\right)^2 + \left(z_b^p\right)^2}}
            \right) + G_y(\cdot)\,y_b^p + G_z(\cdot)\,z_b^p \\
            &\leq G_y(\cdot)\,y_b^p + G_z(\cdot)\,z_b^p + \frac{1}{2} \left(x_b^p\right)^2.
        \end{align}
    \end{subequations}
    Using Young's inequality, we get
    \begin{equation}
            \dot{V}_b\left(\mat{p}_b^p\right) \leq \frac{1}{2} \left(\left(x_b^p\right)^2 + \left(y_b^p\right)^2 + \left(z_b^p\right)^2\right) + \frac{1}{2}\left(G_y(\cdot)^2 + G_z(\cdot)^2\right) = V_b\left(\mat{p}_b^p\right) + \frac{1}{2}\left(G_y(\cdot)^2 + G_z(\cdot)^2\right).
    \end{equation}
    Note that from \eqref{eq:G_y} and \eqref{eq:G_z}, we can conclude that there exist some positive function $\zeta_y(U_{d,1}, \ldots, U_{d,n})$ and $\zeta_z(U_{d,1}, \ldots, U_{d,n})$ such that
    \begin{align}
        \abs{G_y(\cdot)} \leq \zeta_y(\cdot) \left\| \left[\tilde{u}_1, \ldots, \tilde{u}_n, \tilde{\psi}_1, \ldots, \tilde{\psi}_n\right]\T \right\|, \\
        \abs{G_z(\cdot)} \leq \zeta_z(\cdot) \left\| \left[\tilde{u}_1, \ldots, \tilde{u}_n, \tilde{\theta}_1, \ldots, \tilde{\theta}_n\right]\T \right\|. \\
    \end{align}
    Consequently, there exists a class-$\mathcal{K}_{\infty}$ function $\zeta_p(\cdot)$ such that
    \begin{equation}
        \dot{V}_p\left(\mat{p}_b^p\right) \leq V_p\left(\mat{p}_b^p\right) + \zeta_p\left(v_1, \ldots, v_n, w_1, \ldots, w_n, \tilde{u}_1, \ldots, \tilde{u}_n, \tilde{\psi}_1, \ldots, \tilde{\psi}_n, \tilde{\theta}_1, \ldots, \tilde{\theta}_n\right).
    \end{equation}
    Since all the arguments of $\zeta_p(\cdot)$ are forward complete, Corollary 2.11 of \cite{angeli_forward_1999} is satisfied and the barycenter dynamics is forward complete, thus concluding the proof of Lemma~\ref{lemma_1}.

    \subsection{Proof of Lemma \ref{lemma_2}}
    \label{app:lemma_2}
    First, we consider the sway dynamics.
    We take the Lyapunov function candidate $V_v$ from \eqref{eq:V_v} and simplify its derivative by setting $\left[\tilde{\mat{X}}_1\T, \tilde{\mat{X}}_2\T\right] = \mat{0}\T$.
    \begin{equation}
        \dot{V}_v(v_i) = X_v\left(u_{d,i}, u_c\right)\,r_{d,i}\,v_i + Y_v\left(u_{d,i}, u_c\right)\,\left(v_i - v_c\right)\,v_i. \label{eq:V_v_dot_2}
    \end{equation}
    Next, we find an upper bound on $r_{d,i}\,v_i$.
    We substitute from \eqref{eq:r_d}, set $\left[\tilde{\mat{X}}_1\T, \tilde{\mat{X}}_2\T\right] = \mat{0}\T$ and collect all terms that grow linearly with $v_i$ to obtain the following expression
    \begin{align}
        \scale[0.96]{r_{d,i}\,v_i}\, &\scale[0.96]{= \left( v_i\left(1 + \frac{\Delta(\mat{p}_b^p)\,x_b^p}{\Delta(\mat{p}_b^p)^2 + \left(x_b^p\right)^2}\right) \iota(\xi) \frac{1}{n} \sum_{j=1}^n U_j\,\Omega_x(\gamma_j, \theta_p, \chi_j, \psi_p) + \frac{Y_v(u_{d,i}, u_c)}{\sqrt{u_{d,i}^2 + w_i^2}}v_i^2 \right) \cos(\theta_{d,i}) + F_v(u_{d,i}, \theta_{d,i}, u_c, v_c, v_i, w_i, r_i),} \\
        \scale[0.96]{F_v(\cdot)}\, &\scale[0.96]{= \frac{X_v(u_{d,i}, u_c)\,r_i - Y_v(u_{d,i}, u_c)\,v_c}{\sqrt{u_{d,i}^2 + w_i^2}} v_i \, \cos(\theta_{d,i}).}
    \end{align}
    We can bound this expression as
    \begin{subequations}
        \begin{align}
            \abs{r_{d,i}\,v_i} &\leq \frac{2}{n}\abs{v_i}\,\abs{\iota(\xi)}\sum_{j=1}^n\left(\abs{u_j} + \abs{v_j} + \abs{w_j}\right) + \abs{F_v(\cdot)} \\
            &\leq \frac{2}{n}\abs{\iota(\xi)}\,v_i^2 + \frac{2}{n}\abs{v_i}\,\abs{\iota(\xi)}\left(\sum_{j \in \{1, \ldots, n\} \setminus \{i\}}\bigl(\abs{u_j} + \abs{v_j} + \abs{w_j}\bigr) + \abs{u_i} + \abs{w_i}\right) + \abs{F_v(\cdot)},
        \end{align}
    \end{subequations}
    which we can substitute to \eqref{eq:V_v_dot_2} to obtain
    \begin{equation}
        \begin{split}
            \dot{V}_v(v_i) &\leq \left(X_v\left(u_{d,i}, u_c\right)\frac{2}{n}\abs{\iota(\xi)} + Y_v\left(u_{d,i}, u_c\right)\right)v_i^2 + \left(\frac{2}{n}\abs{v_i}\,\abs{\iota(\xi)}\sum_{j \in \{1, \ldots, n\} \setminus \{i\}}\bigl(\abs{u_j} + \abs{v_j} + \abs{w_j}\bigr) + \abs{u_i} + \abs{w_i}\right) \\
            & \quad + \left(\abs{F_v(\cdot)} - Y_v\left(u_{d,i}, u_c\right)\abs{v_c}\right) \abs{v_i}.
        \end{split}
    \end{equation}
    For a sufficiently large $v_i$, the quadratic term will dominate the linear term.
    Therefore, we can conclude that $v_i$ is bounded if 
    \begin{equation}
        X_v\left(u_{d,i}, u_c\right)\frac{2}{n}\abs{\iota(\xi)} + Y_v\left(u_{d,i}, u_c\right) < 0.
    \end{equation}
    Since $Y_v$ is assumed to be always negative, the inequality is satisfied if
    \begin{equation}
        \abs{\iota(\xi)} < \frac{n}{2}\abs{\frac{Y_v\left(u_{d,i}, u_c\right)}{X_v\left(u_{d,i}, u_c\right)}}.
    \end{equation}

    Now, we perform a similar procedure for the heave dynamics.
    We take the Lyapunov function candidate $V_w$ from \eqref{eq:V_w} and simplify its derivative by setting $\left[\tilde{\mat{X}}_1\T, \tilde{\mat{X}}_2\T\right] = \mat{0}\T$.
    \begin{equation}
        \dot{V}_w(w_i) = X_w\left(u_{d,i}, u_c\right)\,q_{d,i}\,w_i + Y_w\left(u_{d,i}, u_c\right)\,\left(w_i - w_c\right)\,w_i + G(\theta_i)\,w_i. \label{eq:V_w_dot}
    \end{equation}
    Next, we find an upper bound on $q_{d,i}\,w_i$.
    We substitute from \eqref{eq:q_d}, set $\left[\tilde{\mat{X}}_1\T, \tilde{\mat{X}}_2\T\right] = \mat{0}\T$ and collect all terms that grow linearly with $w_i$ to obtain the following expression
    \begin{align}
        \scale[0.96]{q_{d,i}\,w_i}\, &\scale[0.96]{= w_i\left(1 + \frac{\Delta(\mat{p}_b^p)\,x_b^p}{\Delta(\mat{p}_b^p)^2 + \left(x_b^p\right)^2}\right) \kappa(\xi) \frac{1}{n} \sum_{j=1}^n U_j\,\Omega_x(\gamma_j, \theta_p, \chi_j, \psi_p) + u_{d,i}\frac{Y_w(u_{d,i}, u_c)}{u_{d,i}^2 + w_i^2}w_i^2 + F(u_{d,i}, u_c, w_c, w_i, q_i),} \\
        \scale[0.96]{F(\cdot)}\, &\scale[0.96]{= u_{d,i}\frac{X_w(u_{d,i}, u_c)\,r_i - Y_w(u_{d,i}, u_c)\,w_c}{\sqrt{u_{d,i}^2 + w_i^2}} w_i.}
    \end{align}
    We can bound this expression as
    \begin{subequations}
        \begin{align}
            \abs{q_{d,i}\,w_i} &\leq \frac{2}{n}\abs{\kappa(\xi)}\,w_i^2 + \frac{2}{n}\abs{w_i}\,\abs{\kappa(\xi)}\left(\sum_{j \in \{1, \ldots, n\} \setminus \{i\}}\bigl(\abs{u_j} + \abs{v_j} + \abs{w_j}\bigr) + \abs{u_i} + \abs{v_i}\right) + \abs{F(\cdot)},
        \end{align}
    \end{subequations}
    which we can substitute to \eqref{eq:V_w_dot} to obtain
    \begin{equation}
        \begin{split}
            \dot{V}_w(w_i) &\leq \left(X_w\left(u_{d,i}, u_c\right)\frac{2}{n}\abs{\kappa(\xi)} + Y_w\left(u_{d,i}, u_c\right)\right)w_i^2 + \left(\frac{2}{n}\abs{w_i}\,\abs{\kappa(\xi)}\sum_{j \in \{1, \ldots, n\} \setminus \{i\}}\bigl(\abs{u_j} + \abs{v_j} + \abs{w_j}\bigr) + \abs{u_i} + \abs{w_i}\right) \\
            & \quad + \left(\abs{F(\cdot)} - Y_w\left(u_{d,i}, u_c\right)\abs{v_c} + \abs{G(\theta_i)}\right) \abs{w_i} + G(\theta_i)\,w_i.
        \end{split}
    \end{equation}
    For a sufficiently large $w_i$, the quadratic term will dominate the linear term.
    Therefore, we can conclude that $w_i$ is bounded if 
    \begin{equation}
        X_w\left(u_{d,i}, u_c\right)\frac{2}{n}\abs{\kappa(\xi)} + Y_w\left(u_{d,i}, u_c\right) < 0.
    \end{equation}
    Since $Y_w$ is assumed to be always negative, the inequality is satisfied if
    \begin{equation}
        \abs{\kappa(\xi)} < \frac{n}{2}\abs{\frac{Y_w\left(u_{d,i}, u_c\right)}{X_w\left(u_{d,i}, u_c\right)}},
    \end{equation}
    which concludes the proof of Lemma \ref{lemma_2}.

    \subsection{Proof of Lemma \ref{lemma_3}}
    \label{app:lemma_3}
    First, we consider the sway dynamics.
    We take the Lyapunov function candidate $V_v$ from \eqref{eq:V_v} and simplify its derivative by setting $\tilde{\mat{X}}_2 = \mat{0}$.
    \begin{equation}
        \dot{V}_v(v_i) = X_v\left(u_{d,i}, u_c\right)\,r_{d,i}\,v_i + Y_v\left(u_{d,i}, u_c\right)\,\left(v_i - v_c\right)\,v_i. \label{eq:V_v_dot_3}
    \end{equation}
    Next, we find an upper bound on $r_{d,i}\,v_i$.
    We substitute from \eqref{eq:r_d}, set $\tilde{\mat{X}}_2 = \mat{0}$ and collect all terms that grow linearly with $v_i$ to obtain the following expression
    \begin{align}
        \begin{split}
            {r_{d,i}\,v_i} &= \scale[1]{\left( v_i\left(1 + \frac{\Delta(\mat{p}_b^p)\,x_b^p}{\Delta(\mat{p}_b^p)^2 + \left(x_b^p\right)^2}\right) \iota(\xi) \frac{1}{n} \sum_{j=1}^n U_j\,\Omega_x(\gamma_j, \theta_p, \chi_j, \psi_p) 
            - \frac{y_b^p\,v_i\,\sum_{j=1}^n\left(\frac{\cos\left(\gamma_{\rm LOS}\right)y_b^p}{\sqrt{\Delta\left(\mat{p}_b^p\right)^2 + \left(y_b^p\right)^2}} + \frac{z_b^p}{\sqrt{\Delta\left(\mat{p}_b^p\right)^2 + \left(z_b^p\right)^2}}\right)}{n\,\Delta(\mat{p}_b^p)\left(\Delta(\mat{p}_b^p)^2 + \left(y_b^p\right)^2\right)} \right.} \\
            &\qquad \left. + \frac{v_i\,\Delta(\mat{p}_b^p)\sum_{j=1}^n\frac{\cos\left(\gamma_{\rm LOS}\right)y_b^p}{\sqrt{\Delta\left(\mat{p}_b^p\right)^2 + \left(y_b^p\right)^2}}}{n\,\left(\Delta(\mat{p}_b^p)^2 + (y_b^p)^2\right)} + \frac{Y_v(u_{d,i}, u_c)}{\sqrt{u_{d,i}^2 + w_i^2}}v_i^2 \right) \cos(\theta_{d,i}) + H_v(u_{d,i}, \theta_{d,i}, u_c, v_c, v_i, w_i, r_i, \mat{p}_b^p, \xi),
        \end{split} \\
        \begin{split}
            H_v(\cdot) & = \left(\left(1 + \frac{\Delta(\mat{p}_b^p)\,x_b^p}{\Delta(\mat{p}_b^p)^2 + \left(x_b^p\right)^2}\right)k_{\xi}\,\iota(\xi)\frac{x_b^p}{\sqrt{1+\left(x_b^p\right)^2}}
            - \frac{y_b^p\,k_{\xi}\,x_b^p}{\sqrt{1+\left(x_b^p\right)^2}\Delta(\mat{p}_b^p)\left(\Delta(\mat{p}_b^p)^2 + \left(y_b^p\right)^2\right)} \right. \\
            &\qquad \left. + \frac{X_v(u_{d,i}, u_c)\,r_i - Y_v(u_{d,i}, u_c)\,v_c}{\sqrt{u_{d,i}^2 + w_i^2}} \right) v_i \, \cos(\theta_{d,i}).
        \end{split}
    \end{align}
    We can bound this expression as
    \begin{subequations}
        \begin{align}
            \abs{r_{d,i}\,v_i} &\leq \left(\frac{2}{n}\abs{\iota(\xi)} + \frac{3}{n\,\Delta(\mat{p}_b^p)}\right)\abs{v_i}\,\sum_{j=1}^n\left(\abs{u_j} + \abs{v_j} + \abs{w_j}\right) + \abs{H_v(\cdot)} \\
            &\leq \left(\frac{2}{n}\abs{\iota(\xi)} + \frac{3}{n\,\Delta(\mat{p}_b^p)}\right)\,v_i^2 + \left(\frac{2}{n}\abs{\iota(\xi)} + \frac{3}{n\,\Delta(\mat{p}_b^p)}\right)\left(\sum_{j \in \{1, \ldots, n\} \setminus \{i\}}\bigl(\abs{u_j} + \abs{v_j} + \abs{w_j}\bigr) + \abs{u_i} + \abs{w_i}\right) \\
            &\quad + \abs{H_v(\cdot)},
        \end{align}
    \end{subequations}
    which we can substitute to \eqref{eq:V_v_dot_3} to obtain
    \begin{equation}
        \begin{split}
            \dot{V}_v(v_i) &\leq \left(X_v\left(u_{d,i}, u_c\right)\left(\frac{2}{n}\abs{\iota(\xi)} + \frac{3}{n\,\Delta(\mat{p}_b^p)}\right) + Y_v\left(u_{d,i}, u_c\right)\right)v_i^2 \\
            & \quad + \left(\frac{2}{n}\abs{\iota(\xi)} + \frac{3}{n\,\Delta(\mat{p}_b^p)}\right)\left(\sum_{j \in \{1, \ldots, n\} \setminus \{i\}}\bigl(\abs{u_j} + \abs{v_j} + \abs{w_j}\bigr) + \abs{u_i} + \abs{w_i}\right) \\
            & \quad + \left(\abs{H_v(\cdot)} - Y_v\left(u_{d,i}, u_c\right)\abs{v_c}\right) \abs{v_i}.
        \end{split}
    \end{equation}
    For a sufficiently large $v_i$, the quadratic term will dominate the linear term.
    Therefore, we can conclude that $v_i$ is bounded if 
    \begin{equation}
        X_v\left(u_{d,i}, u_c\right)\left(\frac{2}{n}\abs{\iota(\xi)} + \frac{3}{n\,\Delta(\mat{p}_b^p)}\right) + Y_v\left(u_{d,i}, u_c\right) < 0.
    \end{equation}
    From the definition of the lookahead distance \eqref{eq:delta}, this condition is satisfied if
    \begin{equation}
        \Delta_0 > \frac{3}{n\abs{\frac{Y_v\left(u_{d,i}, u_c\right)}{X_v\left(u_{d,i}, u_c\right)}} - 2\abs{\iota(\xi)}}.
    \end{equation}

    Now, we perform a similar procedure for the heave dynamics.
    We take the Lyapunov function candidate $V_w$ from \eqref{eq:V_w} and simplify its derivative by setting $\tilde{\mat{X}}_2 = \mat{0}$.
    \begin{equation}
        \dot{V}_w(w_i) = X_w\left(u_{d,i}, u_c\right)\,q_{d,i}\,w_i + Y_w\left(u_{d,i}, u_c\right)\,\left(w_i - w_c\right)\,w_i + G(\theta_i)\,w_i. \label{eq:V_w_dot_2}
    \end{equation}
    Next, we find an upper bound on $q_{d,i}\,w_i$.
    We substitute from \eqref{eq:q_d}, set $\tilde{\mat{X}}_2 = \mat{0}$ and collect all terms that grow linearly with $w_i$ to obtain the following expression
    \begin{align}
        \begin{split}
            {q_{d,i}\,w_i} &= \scale[1]{ w_i\left(1 + \frac{\Delta(\mat{p}_b^p)\,x_b^p}{\Delta(\mat{p}_b^p)^2 + \left(x_b^p\right)^2}\right) \kappa(\xi) \frac{1}{n} \sum_{j=1}^n U_j\,\Omega_x(\gamma_j, \theta_p, \chi_j, \psi_p) 
            - \frac{z_b^p\,w_i\,\sum_{j=1}^n\left(\frac{\cos\left(\gamma_{\rm LOS}\right)y_b^p}{\sqrt{\Delta\left(\mat{p}_b^p\right)^2 + \left(y_b^p\right)^2}} + \frac{z_b^p}{\sqrt{\Delta\left(\mat{p}_b^p\right)^2 + \left(z_b^p\right)^2}}\right)}{n\,\Delta(\mat{p}_b^p)\left(\Delta(\mat{p}_b^p)^2 + \left(z_b^p\right)^2\right)} } \\
            &\quad  + \frac{w_i\,\Delta(\mat{p}_b^p)\sum_{j=1}^n\frac{z_b^p}{\sqrt{\Delta\left(\mat{p}_b^p\right)^2 + \left(z_b^p\right)^2}}}{n\,\left(\Delta(\mat{p}_b^p)^2 + (z_b^p)^2\right)} + u_{d,i}\frac{Y_w(u_{d,i}, u_c)}{{u_{d,i}^2 + w_i^2}}w_i^2 + H_w(u_{d,i}, u_c, v_c, w_i, v_i, q_i, \mat{p}_b^p, \xi),
        \end{split} \\
        \begin{split}
            H_w(\cdot) & = \left(\left(1 + \frac{\Delta(\mat{p}_b^p)\,x_b^p}{\Delta(\mat{p}_b^p)^2 + \left(x_b^p\right)^2}\right)k_{\xi}\,\kappa(\xi)\frac{x_b^p}{\sqrt{1+\left(x_b^p\right)^2}}
            - \frac{y_b^p\,k_{\xi}\,x_b^p}{\sqrt{1+\left(x_b^p\right)^2}\Delta(\mat{p}_b^p)\left(\Delta(\mat{p}_b^p)^2 + \left(y_b^p\right)^2\right)} \right. \\
            &\qquad \left. + u_{d,i}\frac{X_w(u_{d,i}, u_c)\,r_i - Y_w(u_{d,i}, u_c)\,v_c}{{u_{d,i}^2 + w_i^2}} \right) w_i.
        \end{split}
    \end{align}
    We can bound this expression as
    \begin{subequations}
        \begin{align}
            \abs{q_{d,i}\,w_i} &\leq \left(\frac{2}{n}\abs{\kappa(\xi)} + \frac{3}{n\,\Delta(\mat{p}_b^p)}\right)\abs{w_i}\,\sum_{j=1}^n\left(\abs{u_j} + \abs{v_j} + \abs{w_j}\right) + \abs{H_w(\cdot)} \\
            \begin{split}
                &\leq \left(\frac{2}{n}\abs{\kappa(\xi)} + \frac{3}{n\,\Delta(\mat{p}_b^p)}\right)\,w_i^2 + \left(\frac{2}{n}\abs{\kappa(\xi)} + \frac{3}{n\,\Delta(\mat{p}_b^p)}\right)\left(\sum_{j \in \{1, \ldots, n\} \setminus \{i\}}\bigl(\abs{u_j} + \abs{v_j} + \abs{w_j}\bigr) + \abs{u_i} + \abs{w_i}\right) \\
                & \quad + \abs{H_w(\cdot)},    
            \end{split}            
        \end{align}
    \end{subequations}
    which we can substitute to \eqref{eq:V_w_dot_2} to obtain
    \begin{equation}
        \begin{split}
            \dot{V}_w(w_i) &\leq \left(X_w\left(u_{d,i}, u_c\right)\left(\frac{2}{n}\abs{\kappa(\xi)} + \frac{3}{n\,\Delta(\mat{p}_b^p)}\right) + Y_w\left(u_{d,i}, u_c\right)\right)w_i^2 \\
            & \quad + \left(\frac{2}{n}\abs{\kappa(\xi)} + \frac{3}{n\,\Delta(\mat{p}_b^p)}\right)\left(\sum_{j \in \{1, \ldots, n\} \setminus \{i\}}\bigl(\abs{u_j} + \abs{v_j} + \abs{w_j}\bigr) + \abs{u_i} + \abs{w_i}\right) \\
            & \quad + \left(\abs{H_w(\cdot)} - Y_w\left(u_{d,i}, u_c\right)\abs{v_c}\right) \abs{w_i}.
        \end{split}
    \end{equation}
    For a sufficiently large $w_i$, the quadratic term will dominate the linear term.
    Therefore, we can conclude that $w_i$ is bounded if 
    \begin{equation}
        X_w\left(u_{d,i}, u_c\right)\left(\frac{2}{n}\abs{\kappa(\xi)} + \frac{3}{n\,\Delta(\mat{p}_b^p)}\right) + Y_w\left(u_{d,i}, u_c\right) < 0.
    \end{equation}
    From the definition of the lookahead distance \eqref{eq:delta}, this condition is satisfied if
    \begin{equation}
        \Delta_0 > \frac{3}{n\abs{\frac{Y_w\left(u_{d,i}, u_c\right)}{X_w\left(u_{d,i}, u_c\right)}} - 2\abs{\kappa(\xi)}}.
    \end{equation}

\end{document}